\documentclass[journal, onecolumn]{IEEEtran}
\ifCLASSINFOpdf
  \usepackage[pdftex]{graphicx}
  \DeclareGraphicsExtensions{.pdf,.jpeg,.png,.eps}
\else
  \usepackage[dvips]{graphicx}
  \DeclareGraphicsExtensions{.eps}
\fi
\usepackage[cmex10]{amsmath}
\usepackage{amssymb,amsthm}
\usepackage{mathtools, cuted}
\usepackage{lipsum, color, colortbl}
\usepackage[utf8]{inputenc}
\usepackage{cite}
\usepackage{multirow}
\usepackage{array,booktabs}
\newcolumntype{P}[1]{>{\centering\arraybackslash}p{#1}}
\definecolor{Gray}{gray}{0.9}
\usepackage[tight,footnotesize]{subfigure}
\newcommand{\argmax}{\operatornamewithlimits{arg\,max}}
\newcommand{\abs}[1]{\left\lvert{#1}\right\rvert}
\newcommand{\norm}[1]{\left\lVert{#1}\right\rVert}

\newcommand{\nn}{\nonumber}

\newtheorem{theorem}{Theorem}
\newtheorem{prop}{Proposition}
\newtheorem{lemma}{Lemma}
\newtheorem{example}{Example}
\newtheorem{cor}{Corollary}
\newtheorem{remark}{Remark}

\newcommand\blfootnote[1]{%
  \begingroup
  \renewcommand\thefootnote{}\footnote{#1}%
  \addtocounter{footnote}{-1}%
  \endgroup
}
\begin{document}
%
\title{Fundamental Limits of Covert Communication over MIMO AWGN Channel}
%
%
%
%
\author{\IEEEauthorblockN{Amr Abdelaziz and C. Emre Koksal} \\
 \IEEEauthorblockA{Department of Electrical and Computer Engineering \\
The Ohio State University\\
Columbus, Ohio 43201 
}}
\maketitle

\begin{abstract}
Fundamental limits of covert communication have been studied for different models of scalar channels. It was shown that, over $n$ independent channel uses, $\mathcal{O}(\sqrt{n})$ bits can be transmitted reliably over a public channel while achieving an arbitrarily low probability of detection (LPD) by other stations. This result is well known as the square-root law and even to achieve this diminishing rate of covert communication, all existing studies utilized some form of secret shared between the transmitter and the receiver. In this paper, we establish the limits of LPD communication over the MIMO AWGN channel. In particular, using relative entropy as our LPD metric, we study the maximum codebook size for which the transmitter can guarantee reliability and LPD conditions are met. We first show that, the optimal codebook generating input distribution under $\delta$-PD constraint is the zero-mean Gaussian distribution. Then, assuming channel state information (CSI) on only the main channel at the transmitter, we derive the optimal input covariance matrix, hence, establishing scaling laws of the codebook size. We evaluate the codebook scaling rates in the limiting regimes for the number of channel uses (asymptotic block length) and the number of antennas (massive MIMO). We show that, in the asymptotic block-length regime, square-root law still holds for the MIMO AWGN. Meanwhile, in massive MIMO limit, the codebook size, while it scales linearly with $\sqrt{n}$, it scales exponentially with the number of transmitting antennas. Further, we derive equivalent results when \textit{no shared secret} is present. For that scenario, in the massive MIMO limit, higher covert rate up to the non-LPD constrained capacity still can be achieved, yet, with much slower scaling compared to the scenario with shared secret. The practical implication of our result is that, MIMO has the potential to provide a substantial increase in the file sizes that can be covertly communicated subject to a reasonably low delay.
\end{abstract}


\blfootnote{This work was submitted in part to IEEE CNS-2017.}
\blfootnote{This work was in part supported by the National Science Foundation under Grants NSF NeTs 1618566 and 1514260 and Office of Naval Research under Grant N00014-16-1-2253.}

%
\IEEEpeerreviewmaketitle

\section{Introduction}
Conditions for secure communication under a passive eavesdropping attack fall in two broad categories: 1) \textit{low probability of intercept (LPI)}. 2) \textit{low probability of detection (LPD)}. Communication with LPI requires the message exchanged by two legitimate parties to be kept secret from an illegitimate adversary. Meanwhile, LPD constrained communication is more restrictive as it requires the adversary to be unable to decide whether communication between legitimate parties has taken place. 
Fundamental limits of LPD constrained communication over scalar AWGN has been established in \cite{bash2013limits} where the square-root law for LPD communication was established. Assuming  a shared secret of sufficient length between transmitter and receiver, square-root law states that, over $n$ independent channel uses of an AWGN channel, transmitter can send $\mathcal{O}(\sqrt{n})$ bits \textit{reliably} to the receiver while keeping arbitrary low probability of detection at the adversary. In this paper, we study the fundamental limits of communication with LPD over MIMO AWGN channels. 

Consider the scenario in which a transmitter (Alice) wishes to communicate to a receiver (Bob) while being undetected by a passive adversary (Willie) when all nodes are equipped with multiple antennas. To that end, Alice wish to generate a codebook that satisfies both reliability, in terms of low error probability $\epsilon$, over her channel to Bob and, in the same time, ensures, a certain maximum PD, namely $\delta$, at Willie. Denote the maximum possible size of such codebook by $K_n(\delta,\epsilon)$. In this paper, we are interested in establishing the fundamental limits of $K_n(\delta,\epsilon)$ in the asymptotic length length regime and in the limit of large number of transmitting antenna.  First we show that, the maximum codebook size is attained when the codebook is generated according to zero mean circular symmetric complex Gaussian distribution. We establish this result building upon the the Principle Minimum Relative Entropy \cite{woodbury2011minimum} and Information Projection \cite{csiszar2003information}.

Some of our findings can be summarized as follows. For an isotropic Willie channel, we show that Alice can transmit $\mathcal{O}(N \sqrt{n/M})$ bits reliably in $n$ independent channel uses, where $N$ and $M$ are the number of active eigenmodes of Bob and Willie channels, respectively. Further, we evaluate $\delta$-PD rates in the limiting regimes for the number of channel uses (asymptotic block length) and the number of antennas (massive MIMO). We show that, while the square-root law still holds for the MIMO AWGN, the number of bits that can be transmitted covertly scales exponentially with the number of transmitting antennas. More precisely, for a unit rank MIMO channel, we show that $K_n(\delta,\epsilon)$ scales as $\sqrt{\dfrac{n}{K^2 N_a}}(1+\dfrac{c}{\sqrt{n}})^{(N_a-2)/2}$ where $N_a$ is the number of transmitting antennas, $K$ is a universal constant and $c$ is constant independent on $n$ and $N_a$. Further, we derive the scaling of $K_n(\delta,\epsilon)$ \textit{with no shared secret} between Alice and Bob. In particular, we show that achieving better covert rate is a resource arm race between Alice, Bob and Willie. Alice can transmit $\mathcal{O}(N/M)$ bits reliably in $n$ independent channel uses, i.e., the covert rate is in the order of the ratio between active eigenmodes of both channels.  The practical implication of our findings is that, MIMO has the potential to provide a substantial increase in the file sizes that can be covertly communicated subject to a reasonably low delay. The results obtained in this paper are summarized in Table \ref{table:summary}\footnote{$\theta$ is the angle between right singular vectors of main and adversary channels in the unit rank channel model.}. 
\begin{table*}
\def\arraystretch{2}\tabcolsep=10pt
\centering
\caption{Summary of Results}
\centering
 \begin{tabular}{|| p{2.2cm} | p{3cm} | p{3cm} | p{1cm} | p{4cm} ||} 
 \rowcolor{Gray}
 \hline\hline
 \textbf{Result}  & \textbf{Main Channel}  & \textbf{Adversary Channel} & \textbf{Shared Secret} & $K_n(\delta,\epsilon)$ \textbf{Scales Like} \\ [0.5ex] 
 \hline\hline
\textit{Theorems \ref{thm:main_general}\&\ref{thm:srl_mimo}}     & Deterministic         & Bounded Spectral norm    & Yes  &  $N\sqrt{n/M}$             \\ 
\hline
\textit{Theorem \ref{thm:rank1}}     &   Deterministic and of Unit Rank       &    Deterministic and of Unit Rank   &       Yes     & $\sqrt{n}/ \cos^2(\theta)$  \\
 \hline
 \textit{Theorems \ref{thm:main_nosk}\&\ref{thm:srl_mimo_nosk}}    &     Deterministic     &    Bounded Spectral norm   &       No         & $N/M$ \\
 \hline 
 \textit{Theorem \ref{thm:rank1_nosk}}  &      Deterministic and of Unit Rank        &     Deterministic and of Unit Rank        &         No       &     $1/ cos^2(\theta)$         \\
 \hline
 \textit{Theorem \ref{thm:limit_main}}    &    Deterministic and of Unit Rank   & Unit Rank chosen uniformly at random             &    Yes    &  $\sqrt{\dfrac{n}{K^2 N_a}}(1+\dfrac{c}{\sqrt{n}})^{(N_a-2)/2}$ where $c$ is constant independent on $n$ and $N_a$. \\ 
  \hline
 \textit{Theorem \ref{thm:limit_main_nosk}} & Deterministic and of Unit Rank &  Unit Rank chosen uniformly at random & No &  $\dfrac{1}{K \sqrt{ N_a}}(1+\dfrac{c}{n})^{(N_a-2)/2}$ where $c$ is constant independent on $n$ and $N_a$.  \\ [0.5ex]
 \hline
\end{tabular}
\label{table:summary}
\end{table*}

The contributions of this work can be summarized as follows:
\begin{itemize}
\item  Using the Principle Minimum Relative Entropy \cite{woodbury2011minimum} and Information Projection \cite{csiszar2003information}, we show that the $K_n(\delta ,\epsilon)$ is achievable when the codebook is generated according to zero mean complex Gaussian distribution in MIMO AWGN channels.
\item  With the availability of only the main CSI to Alice, we evaluate the optimal input covariance matrix under the assumption that Willie channel satisfies a bounded spectral norm constraint \cite{schaefer2015secrecy,Abde1706_Compound}. Singular value decomposition (SVD) precoding is shown to be the optimal signaling strategy and the optimal water-filling strategy is also provided.
\item We evaluate the block-length and massive MIMO asymptotics for $K_n(\delta ,\epsilon)$. We show that, while the square-root law cannot be avoided, $K_n(\delta ,\epsilon)$ scales exponentially with the number of antennas. Thus, MIMO has the potential to provide a substantial increase in the file sizes that can be covertly communicated subject to a reasonably low delay.
\item We evaluate scaling laws of $K_n(\delta ,\epsilon)$ when there is no shared secret between Alice and Bob in both limits of large block length and massive MIMO. 
\end{itemize}

\textbf{Related Work.} Fundamental limits of covert communication have been studied in literature for different models of scalar channels. In \cite{reliable_deniable}, LPD communication over the binary symmetric channel was considered. It was shown that, square-root law holds for the binary symmetric channel, yet, without requiring a shared secret between Alice and Bob when Willie channel is significantly noisier. Further, it was shown that Alice achieves a non-diminishing LPD rate, exploiting Willie's uncertainty about his own channel transition probabilities. Recently in \cite{bloch2016covert}, LPD communication was studied from a resolvability prespective for the discrete memoryless channel (DMC). Therein, a trade-off between the secret length and asymmetries between Bob and Willie channels has been studied. Later in \cite{wang2016fundamental}, the exact capacity (using relative entropy instead of total variation distance as LPD measure) of DMC and AWGN have been characterized. For a detailed summary of the recent results for different channel models on the relationship between secret key length, LPD security metric and achievable LPD rate, readers may refer to Table II in \cite{reliable_deniable}. LPD communication over MIMO fading channel was first studied in \cite{hero2003secure}. Under different assumption of CSI availability, the author derived the average power that satisfies the LPD requirement. However, the authors did not obtain the square-root law, since the focus was not on the achievable rates of reliable LPD communication. Recently in \cite{lee2014achieving}, LPD communication with multiple antennas at both Alice and Bob is considered when Willie has only a single antenna over Rayleigh fading channel. An approximation to the LPD constrained rate when Willie employs a radiometer detector and has uncertainty about his noise variance was presented. However, a full characterization of the capacity of MIMO channel with LPD constraint was not established.

\textit{ Despite not explicitly stated, the assumption of keeping the codebook generated by Alice secret from Willie (or at least a secret of sufficient length \cite{bash2013limits,bloch2016covert}) is common in all aforementioned studies of covert communication}. Without this assumption, LPD condition cannot be met along with arbitrarily low probability of error at Bob. This is because, when Willie is informed about the codebook, he can decode the message using the same decoding strategy as that of Bob \cite{bash2013limits}. Only in \cite{reliable_deniable}, square-root law was obtained over binary symmetric channel without this assumption when Willie channel is significantly noisier than that of Bob, i.e., when there is a positive secrecy rate over the underlying wiretap channel. Despite the availability of the codebook at Willie, \cite{reliable_deniable} uses the total variation distance as the LPD metric.

In short, the square root law is shown to be a fundamental upper limitation that cannot be overcome unless the attack model is relaxed to cases such as the lack of CSI or the lack of the knowledge of when the session starts at Willie. Here, we do not make such assumptions on Willie and solely take advantage of increasing spatial dimension via the use of MIMO.

\section{System Model and Problem Statement}
\label{sec:model}
In the rest of this paper we use boldface uppercase letters for vectors/matrices. Meanwhile, $(.)^{*}$ denotes conjugate of complex number, $(.)^{\dagger}$ denotes conjugate transpose, $\mathbf{I}_N$ denotes identity matrix of size $N$, $\mathbf{tr}(.)$ denotes matrix trace operator, $\abs{\mathbf{A}}$ denotes the determinant of matrix $A$ and $\mathbf{1}_{m \times n}$ denotes a $m \times n$ matrix of all 1's. We say $\mathbf{A} \succeq \mathbf{B}$ when the difference $\mathbf{A} - \mathbf{B}$ is positive semi-definite. The mutual information between two random variables $x$ and $y$ denoted by $\mathcal{I}(x;y)$ while $\varliminf$ denotes the limit inferior. We use the standard order notation $f(n)=\mathcal{O}(g(n))$ to denote an upper bound on $f(n)$ that is asymptotically tight, i.e., there exist a constant $m$ and $n_0>0$ such that $0 \leq f(n) \leq m g(n)$ for all $n>n_0$. 

\subsection{Communication Model}
We consider the MIMO channel scenario in which a transmitter, Alice, with $N_a \geq 1$ antennas aims to communicate with a receiver, Bob, having $N_b \geq 1$ antennas without being detected by a passive adversary, Willie, equipped with $N_w \geq 1$ antennas. The discrete baseband equivalent channel for the signal $\mathbf{y}$ and $\mathbf{z}$, received by Bob and Willie, respectively, can be written as:
\begin{align}
\label{eq:inout_security}
&\mathbf{y} = \mathbf{H}_b   \mathbf{x} + \mathbf{e}_b, \nonumber \\
&\mathbf{z} = \mathbf{H}_w   \mathbf{x} + \mathbf{e}_w,
\end{align} 
where  $\mathbf{x} \in \mathbb{C}^{N_a\times 1}$ is the transmitted signal vector constrained by an average power constraint $\mathbb{E}[\mathbf{tr}(\mathbf{x}\mathbf{x}^{\dagger})] \leq P$. Also, ${\mathbf{H}_b} \in \mathbb{C}^{N_b \times N_a}$ and $\mathbf{H}_w  \in \mathbb{C}^{N_w \times N_a}$ are the channel coefficient matrices for Alice-Bob and Alice-Willie channels respectively. Throughout this paper, unless otherwise noted, $\mathbf{H}_b$ and $\mathbf{H}_w$ are assumed deterministic, also, we assume that $\mathbf{H}_b$ is known to all parties, meanwhile, $\mathbf{H}_w$ is known only to Willie. We define $N \triangleq \min\{N_a,N_b\}$ and $M\triangleq \min\{N_a,N_w\}$. Finally, $\mathbf{e}_b \in \mathbb{C}^{N_b\times 1}$ and  $\mathbf{e}_w \in \mathbb{C}^{N_w\times 1}$ are an independent zero mean circular symmetric complex Gaussian random vectors for both destination and adversary channels respectively, where, $\mathbf{e}_b\sim \mathcal{CN}(0,\sigma_b^2 \mathbf{I}_{N_b})$ and $\mathbf{e}_w\sim \mathcal{CN}(0,\sigma_e^2 \mathbf{I}_{N_w})$. 

We further assume $\mathbf{H}_w$ lies in the set of matrices with bounded spectral norms:
\begin{align}
\label{eq:willie_compound}
\mathcal{S}_w &= \left\{\mathbf{H}_w : \norm{\mathbf{H}_w}_{op} \leq \sqrt{\gamma_w}\right\} \nn \\
&=\left\{\mathbf{W}_w \triangleq \mathbf{H}_w^{\dagger}\mathbf{H}_w: \norm{\mathbf{W}_w}_{op} \leq \gamma_w\right\},
\end{align} 
where $\norm{A}_{op}$ is the operator (spectral) norm of $\mathbf{A}$, i.e., the maximum eigenvalue of $\mathbf{A}$. The set $\mathcal{S}_w$ incorporates all possible $\mathbf{W}_w$ that is less than or equal to $\gamma_w \hat{\mathbf{I}}$ (in positive semi-definite sense) with no restriction on its eigenvectors, where $\hat{\mathbf{I}}$ is diagonal matrix with the first $M$ diagonal elements equal to $1$ while the rest $N_a - M$ elements of the diagonal are zeros. Observe that, $\norm{\mathbf{W}_w}_{op}$ represents the largest possible power gain of Willie channel. Unless otherwise noted, throughout this paper we will assume that $\mathbf{H}_w \in \mathcal{S}_w$.

\subsection{Problem Statement}
\label{sec:problem}

Our objective is to establish the fundamental limits of reliable transmission over Alice to Bob MIMO channel, constrained by the low detection probability at Willie. Scalar AWGN channel channel have been studied in \cite{wang2016fundamental}, we use a formulation that follows closely the one used therein while taking into consideration the vector nature of the MIMO channel. Alice employs a stochastic encoder with blocklength\footnote{Note that, when Alice has $nN_a$ bits to transmit, two alternative options are available for her. Either she splits the incoming stream into $N_a$ streams of $n$ bits each and use each stream to select one from $2^{n}$ messages for each single antenna, or, use the entire $n N_a$ bits to choose from $2^{n N_a}$ message. The latter of these alternatives provides a gain factor of $N_a$ in the error exponent, of course, in the expense of much greater complexity \cite{gallager1968information,telatar1999capacity}. However, in the restrictive LPD scenario, Alice would choose the latter alternative as to achieve the best decoding performance at Bob.} $n N_a$, where $n$ is the number of channel uses, for message set $\mathcal{M}$ consists of:
\begin{enumerate}
\item An encoder $\mathcal{M} \mapsto \mathbb{C}^{nN_a}$, $m \mapsto \mathbf{x}^{n}$ where $\mathbf{x}\in \mathbb{C}^{N_a}$.
\item A decoder $\mathbb{C}^{nN_b} \mapsto \mathcal{M}$, $\mathbf{y}^{n} \mapsto \hat{m}$ where $\mathbf{y}\in \mathbb{C}^{N_b}$.
\end{enumerate} 
Alice chooses a message $M$ from $\mathcal{M}$ uniformly at random to transmit to Bob. Let us denote by $\mathcal{H}_0$ the null hypothesis under which Alice is not transmitting and denote by $\mathbb{P}_0$ the probability distribution of Willie's observation under the null hypothesis. Conversely, let $\mathcal{H}_1$ be the true hypothesis under which Alice is transmitting her chosen message $M$ and let $\mathbb{P}_1$ be the probability distribution of Willie's observation under the true hypothesis. Further, define type \textbf{I} error $\alpha$ to be the probability of mistakenly accepting $\mathcal{H}_1$ and type \textbf{II} error $\beta$ to be the probability of mistakenly accepting $\mathcal{H}_0$. For the optimal hypothesis test generated by Willie we have \cite{lehmann2006testing}
\begin{align}
\label{eq:sum_err}
\alpha + \beta = 1 - \mathcal{V}(\mathbb{P}_0,\mathbb{P}_1),
\end{align} 
where $\mathcal{V}(\mathbb{P}_0,\mathbb{P}_1)$ the total variation distance between $\mathbb{P}_0$ and $\mathbb{P}_1$ and is given by
\begin{align}
\label{eq:VT}
\mathcal{V}(\mathbb{P}_0,\mathbb{P}_1) = \dfrac{1}{2}\norm{p_0(x)-p_1(x)}_1,
\end{align}
where $p_0(x)$ and $p_1(x)$ are, respectively, the densities of $\mathbb{P}_0$ and $\mathbb{P}_1$ and $\norm{.}_1$ is the $\mathcal{L}_1$ norm. The variation distance between $\mathbb{P}_0$ and $\mathbb{P}_1$ is related to the Kullback–Leibler Divergence (relative entropy) by the well known Pinsker's inequality \cite{cover2012elements}:
\begin{align}
\label{eq:Pins_Ineq}
\mathcal{V}(\mathbb{P}_0,\mathbb{P}_1) \leq \sqrt{\dfrac{1}{2} \mathcal{D}(\mathbb{P}_0 \parallel \mathbb{P}_1)}
\end{align} 
where
\begin{align}
\label{eq:rel_ent}
 \mathcal{D}(\mathbb{P}_0 \parallel \mathbb{P}_1) = \mathbb{E}_{\mathbb{P}_0} \left[\log \mathbb{P}_0 - \log \mathbb{P}_1\right].
\end{align} 
Note that, since the channel is memoryless, across $n$ independent channel uses, we have
\begin{align}
\mathcal{D} \left(\mathbb{P}_0^n \parallel \mathbb{P}_1^n \right) = n \mathcal{D} \left(\mathbb{P}_0 \parallel \mathbb{P}_1 \right)
\end{align}
by the chain rule of relative entropy. Accordingly, for Alice to guarantee a low detection probability at Willie's optimal detector, she needs to bound $\mathcal{V}(\mathbb{P}_0^n,\mathbb{P}_1^n)$ above by some $\delta$ chosen according to the desired probability of detection. Consequently, she ensure that the sum of error probabilities at Willie is bounded as $\alpha + \beta \geq 1-\delta $. Using (\ref{eq:Pins_Ineq}), Alice can achieve her goal by designing her signaling strategy (based on the amount of information available) subject to
\begin{align}
\label{eq:metric}
\mathcal{D}(\mathbb{P}_0 \parallel \mathbb{P}_1) \leq \dfrac{2 \delta^2}{n}.
\end{align} 
Throughout this paper, we adopt \eqref{eq:metric} as our LPD metric. Thus, the input distribution used by Alice to generate the codebook has to satisfy \eqref{eq:metric}. As in \cite{wang2016fundamental}, our goal is to find the maximum value of $\log \abs{\mathcal{M}}$ for which a random codebook of length $nN_a$ exists and satisfies \eqref{eq:metric} and whose average probability of error is at most $\epsilon$. We denote this maximum by $K_n(\delta,\epsilon)$ and we define
\begin{align}
\label{eq:L}
L \triangleq \lim_{\epsilon \downarrow 0} \varliminf_{n \rightarrow \infty} \dfrac{K_n(\delta,\epsilon)}{\sqrt{2n \delta^2}}.
\end{align}
Note that $L$ has unit $\sqrt{nats}$. We are interested in the characterization of $L$ under different conditions of Bob and Willie channels in order to derive scaling laws for the number of covert bits over MIMO AWGN channel. We first give the following Proposition which provides a general expression for $L$ by extending Theorem 1 in \cite{wang2016fundamental} to the MIMO AWGN channel with infinite input and output alphabet.
\begin{prop}
\label{prop:l_mimo}
For the considered MIMO AWGN channel, 
\begin{align}
\label{eq:l_mimo_prop}
L = &\max_{\substack{\left\{f_n(\mathbf{x})\right\} \\ \mathbf{tr}\left(\mathbb{E}_n[\mathbf{x}\mathbf{x}^{\dagger}]\right) \leq P}} \varliminf_{n \rightarrow \infty} \sqrt{\dfrac{n}{2\delta^2}} \;\;\; \mathcal{I}(f_n(\mathbf{x}),f_n(\mathbf{y}))\nn \\
&\text{Subject to:  } \mathcal{D}(\mathbb{P}_0^n \parallel \mathbb{P}_1^n)- 2 \delta^2 \leq 0
\end{align}
where $\left\{f_n(\mathbf{x})\right\}$ is a sequence of input distributions over $\mathbb{C}^{N_a}$ and $\mathbb{E}_n[\cdot]$ denotes the expectation with respect to $f_n(\mathbf{x})$.
\end{prop}
Before we give the proof of Proposition \ref{prop:l_mimo}, we would like to highlight why the second moment constraint on the input signal is meaningful in our formulation. It was explicitly stated in \cite{wang2016fundamental} that, an average power constraint on the input signal is to be superseded by the LPD constraint. The reason is that, the LPD constraint requires the average power to tend to zero as the block length tends to $\infty$. However, unlike the single antenna setting, over a MIMO channel, there exist scenarios in which the LPD constraint can be met without requiring the Alice to reduce her power. In the sequel, we will discuss such scenarios in which the power constraint remains active.

\begin{proof}
First, using the encoder/decoder structure described above, we see that the converse part of Theorem 1 in \cite{wang2016fundamental} can be directly applied here. Meanwhile, the achievability  part there was derived based on the finiteness of input and output alphabet. It was not generalized to the continuous alphabet input over scalar AWGN channel. Rather, the achievability over AWGN channel was shown for Gaussian distributed input in Theorem 5. Here, we argue that, showing achievability for Gaussian distributed input is sufficient and, hence, we give achievability proof in Appendix \ref{app:l_mimo_prop} that follow closely the proof of Theorem 5 in \cite{wang2016fundamental}. Unlike the non LPD constrained capacity which attains its maximum when the underlying input distribution is zero mean complex Gaussian, it is not straightforward to infer what input distribution is optimal. However, using the Principle Minimum Relative Entropy \cite{woodbury2011minimum} and Information Projection \cite{csiszar2003information}, we verify that, the distribution $\mathbb{P}_1$ that minimizes $\mathcal{D}(\mathbb{P}_0 \parallel \mathbb{P}_1)$ is the zero mean circularly symmetric complex Gaussian distribution.
\end{proof}
Further, we provide a more convenient expression for $L$ in the following Theorem which provides an extension of Corollary 1 in \cite{wang2016fundamental} to the MIMO AWGN channel.
\begin{theorem}
\label{thm:l_mimo}
For the considered MIMO AWGN channel, 
\begin{align}
\label{eq:l_mimo}
L = & \lim_{n \rightarrow \infty} \sqrt{\dfrac{n}{2\delta^2}}  \max_{\substack{f_n(\mathbf{x}) \\ \mathbf{tr}\left(\mathbb{E}_n[\mathbf{x}\mathbf{x}^{\dagger}]\right) \leq P}} \mathcal{I}(f_n(\mathbf{x}),f_n(\mathbf{y}))\nn \\
&\text{Subject to:  } \mathcal{D}(\mathbb{P}_0^n \parallel \mathbb{P}_1^n)- 2 \delta^2 \leq 0
\end{align}
where $f_n(\mathbf{x})$ is the input distribution over $\mathbb{C}^{N_a}$ and $\mathbb{E}_n[\cdot]$ denotes the expectation with respect to $f_n(\mathbf{x})$.
\end{theorem}
\begin{proof}
The proof is given in Appendix \ref{app:l_mimo}.
\end{proof}
Now, since we now know that zero mean circular symmetric complex Gaussian input distribution is optimal, the only remaining task is to characterize the covariance matrix, $\mathbf{Q} = \mathbb{E}\left[\mathbf{x}\mathbf{x}^{\dagger}\right]$, of the optimal input distribution. Accordingly, \eqref{eq:l_mimo} can be rewritten as:
\begin{align}
\label{eq:l_MIMO_proplem}
L =  \lim_{n \rightarrow \infty}  \sqrt{\dfrac{n}{2\delta^2}} &\max_{\substack{\mathbf{Q} \succeq \mathbf{0} \\ \mathbf{tr}(\mathbf{Q}) \leq P}} \log \abs{\mathbf{I}_{N_a} + \dfrac{\mathbf{W}_b\mathbf{Q}}{\sigma_b^2}}  \\
&\text{Subject to:  } \mathcal{D}(\mathbb{P}_0^n \parallel \mathbb{P}_1^n)- 2 \delta^2 \leq 0, \nn
\end{align}  
where $\mathbf{W}_b \triangleq \mathbf{H}_b^{\dagger}\mathbf{H}_b$. Further, we can evaluate the relative entropy at Willie as follows (see Appendix \ref{app:KL} for detailed derivation):
\begin{align}
\label{eq:D_gen_cont}
\mathcal{D} \left(\mathbb{P}_0 \parallel \mathbb{P}_1 \right) = &\log \abs{\dfrac{1}{ \sigma_w^2 } \mathbf{H}_w \mathbf{Q}\mathbf{H}_w^{\dagger} + \mathbf{I}_{N_w}}  \nonumber  \\
&+ \mathbf{tr}\left\{\left[\dfrac{1}{\sigma_w^2 } \mathbf{H}_w \mathbf{Q}\mathbf{H}_w^{\dagger} + \mathbf{I}_{N_w}\right]^{-1}\right\} - N_w.
\end{align}
In this paper, we are mainly concerned with characterizing $L$ when Alice knows only $\mathbf{H}_b$. To that end, let us define:
\begin{align}
\label{eq:eq_MIMO_proplem}
C_{pd}(\delta) \triangleq  &\max_{\substack{\mathbf{Q} \succeq \mathbf{0} \\ \mathbf{tr}(\mathbf{Q}) \leq P}} \log \abs{\mathbf{I}_{N_a} + \dfrac{\mathbf{W}_b\mathbf{Q}}{\sigma_b^2}}  \\
&\text{Subject to:    } \mathcal{D}(\mathbb{P}_0^n \parallel \mathbb{P}_1^n)- 2 \delta^2 \leq 0. \nn
\end{align}  
Clearly, $L = \lim_{n\rightarrow \infty} \sqrt{\dfrac{n}{2\delta^2}} C_{pd}(\delta)$. In what follows, we characterize $C_{pd}(\delta)$ and, hence, $L$ under different models of $\mathbf{H}_b$ and $\mathbf{H}_w$.

\begin{remark}
Observe that, since Bob and Willie channels are different, Willies does not observe the same channel output as Bob. Hence, there exist situations in which $\mathcal{D}(\mathbb{P}_0^n \parallel \mathbb{P}_1^n)$ does not increase without bound as $n$ tends to infinity. In the next Section, we provide some examples. 
\end{remark}

\section{Motivating Examples}
\label{sec:example}
Consider the scenario in which both of Willie's and Bob's channels are of unit rank. Accordingly, we can write $\mathbf{H}_{\circ} = \lambda_{\circ}v_{\circ}u_{\circ}^{\dagger}$, where $v_{\circ} \in \mathbb{C}^{N_{\circ}}$ and $u_{\circ} \in \mathbb{C}^{N_{a}}$ are the left and right singular vectors of $\mathbf{H}_{\circ}$ where the subscript $\circ \in \{e,b\}$ used to denote Bob and Willie channels respectively.

Under the above settings, consider the scenario in which Alice has a prior (non-causal) knowledge about both channels. Alice task is to find $\mathbf{Q}_*$ that solve (\ref{eq:eq_MIMO_proplem}). Note that, since both channels are of unit rank, so is $\mathbf{Q}_*$ and it can be written as $\mathbf{Q}_* = P_{th} q_*q_*^{\dagger}$ where $P_{th} \leq P$ is the power threshold above which she will be detected by Willie. Now suppose that Alice choose $q_*$ to be the solution of the following optimization problem:
\begin{align}
\max_{\substack{{q} \\ \norm{q}=1}} <q^{\dagger},\mathbf{u}_b> \nonumber \\
\text{Subject to} <q^{\dagger},\mathbf{u}_w> = 0,
\end{align} 
whose solution is given by
\begin{align}
q_{*}= \dfrac{\left[\mathbf{I}-\mathbf{u}_w\mathbf{u}_w^{\dagger}\right]\mathbf{u}_b}{\norm{\left[\mathbf{I}-\mathbf{u}_w\mathbf{u}_w^{\dagger}\right] \mathbf{u}_b}}.
\end{align}
The beamforming direction $q_*$ is known as null steering (NS) beamforming \cite{null_Steering}, that is, transmission in the direction orthogonal to Willie's direction while maintaining the maximum possible gain in the direction of Bob. Recall that Willies channel is of unit rank and is in the direction $\mathbf{u}_w$, thus, the choice of $\mathbf{Q} = \mathbf{Q}_*$ implies that $\mathbf{H}_w\mathbf{Q}_* \mathbf{H}_w^{\dagger} = \mathbf{0}$. Accordingly, $\Sigma_1 = \Sigma_0$, i.e, Willie is kept completely ignorant by observing absolutely no power from Alice's transmission. More precisely, $\mathcal{D} \left(\mathbb{P}_0^n \parallel \mathbb{P}_1^n \right) = 0$. However, this doesn't mean that Alice can communicate at the full rate to Bob as if Willie was not observing, rather, the LPD constraint forced Alice to sacrifice some of its power to keep Willie oblivious of their transmission. More precisely, the effective power seen by Bob scales down with cosine the angle between $\mathbf{u}_b$ and $\mathbf{u}_w$. In the special case when $<\mathbf{u}_b,\mathbf{u}_w>=0$, Alice communicate at the full rate to Bob without being detected by Willie. In addition, the codebook between Alice and Bob need not to be kept secret from Willie. That is because the power observed at Willie from Alice transmission is, in fact, zero.
 \begin{figure}[ht]
\centering
\includegraphics[width = 2.2 in, height = 2.4 in]{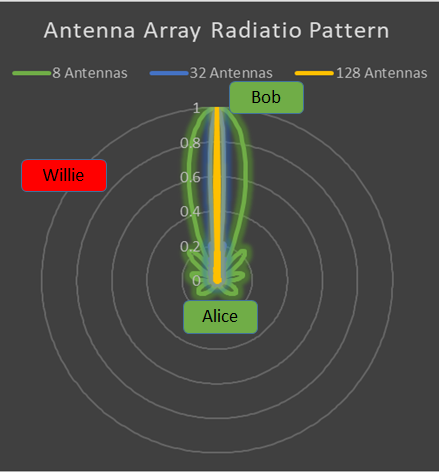}
\caption{Radiation pattern as a function of the number of transmitting antennas. When number of antennas gets large, $<\mathbf{u}_b,\mathbf{u}_w> \rightarrow 0$.  
\label{fig:example_fig}}
\end{figure}

\section{$C_{pd}(\delta)$ with Secret Codebook}
\label{sec:main_informed}
With uncertainty about Willie's channel, $\mathbf{H}_w \in \mathcal{S}_w$, it is intuitive to think that Alice should design her signaling strategy against the worst (stronger) possible Willie channel. We first derive the worst case Willie channel, then, we establish the saddle point property of the considered class of channels in the form of $\min \max = \max \min$, where the maximum is taken over all admissible input covariance matrices and the minimum is over all   $\mathbf{H}_w \in \mathcal{S}_w$. Thus, we show that $C_{pd}(\delta)$ equals to the $C_{pd}(\delta)$ evaluated at the worst possible $\mathbf{H}_w$.

\subsection{Worst Willie Channel and Saddle Point Property}
To characterize $C_{pd}(\delta)$ when $\mathbf{H}_w \in \mathcal{S}_w$, we need first to establish the worst case $C_{pd}(\delta)$ denoted by $C_{pd}^w(\delta)$. Suppose we have obtained $C_{pd}(\delta)$ for every possible state of $\mathbf{H}_w$, then, $C_{pd}^w(\delta)$ is the minimum $C_{pd}(\delta)$ over all possible state of $\mathbf{H}_w$. First, let us define 
\begin{align}
\label{eq:lpd_rate}
\mathcal{R}(\mathbf{W}_w, \mathbf{Q} ,\delta) = &\log \abs{\mathbf{I}_{N_a} + \dfrac{\mathbf{W}_b\mathbf{Q}}{\sigma_b^2}}.
\end{align}   
We give $C_{pd}^w(\delta)$ in the following proposition. 
\begin{prop}
\label{prop:W_ww} 
Consider the class of channels in (\ref{eq:willie_compound}), for any $\mathbf{Q} \succcurlyeq \mathbf{0}$ satisfies $\mathbf{tr}\{\mathbf{Q}\} \leq P$ and $\mathbf{W}_w \in \mathcal{S}_w$ we have:
\begin{align}
C_{pd}^w(\delta) &= \min_{\mathbf{W}_w \in \mathcal{S}_w}  \max_{\substack{\mathbf{Q} \succeq \mathbf{0} \\ \mathbf{tr}(\mathbf{Q}) \leq P}} \mathcal{R}(\mathbf{W}_w, \mathbf{Q} ,\delta) \nn \\
&\,\,\,\,\,\text{Subject to:  } \mathcal{D}(\mathbb{P}_0^n \parallel \mathbb{P}_1^n)- 2 \delta^2 \leq 0 \nn \\
& = \max_{\substack{\mathbf{Q} \succeq \mathbf{0} \\ \mathbf{tr}(\mathbf{Q}) \leq P}} \mathcal{R}( \gamma_w \hat{\mathbf{I}}, \mathbf{Q} ,\delta) \nn \\
&\,\,\,\,\,\text{Subject to:  } \mathcal{D}(\mathbb{P}_0^n \parallel \mathbb{P}_1^n)- 2 \delta^2 \leq 0
\end{align}
i.e., the worst Willie channel is isotropic.
\end{prop}
\begin{proof} See Appendix \ref{app:W_ww}.
\end{proof}
Proposition \ref{prop:W_ww} establishes $C_{pd}^w(\delta)$. The following proposition proves that $C_{pd}(\delta) = C_{pd}^w(\delta)$ by establishing the saddle point property of the considered class of channels.

\begin{prop}
\label{prop:saddle}
(Saddle Point Property.) Consider the class of channels in (\ref{eq:willie_compound}), for any $\mathbf{Q} \succcurlyeq \mathbf{0}$ satisfies $\mathbf{tr}\{\mathbf{Q}\} \leq P$ and $\mathbf{W}_w \in \mathcal{S}_w$ we have:
\begin{align}
\label{eq:worst_lpd_cap}
C_{pd}(\delta) &= \min_{\mathbf{W}_w \in \mathcal{S}_w}  \max_{\substack{\mathbf{Q} \succeq \mathbf{0} \\ \mathbf{tr}(\mathbf{Q}) \leq P}} \mathcal{R}(\mathbf{W}_w, \mathbf{Q} ,\delta) \nn \\
&\,\,\,\,\,\text{Subject to:  } \mathcal{D}(\mathbb{P}_0^n \parallel \mathbb{P}_1^n)- 2 \delta^2 \leq 0 \nn \\
& = \max_{\substack{\mathbf{Q} \succeq \mathbf{0} \\ \mathbf{tr}(\mathbf{Q}) \leq P}}  \min_{\mathbf{W}_w \in \mathcal{S}_w}  \mathcal{R}( \mathbf{W}_w , \mathbf{Q} ,\delta) \nn \\
&\,\,\,\,\,\text{Subject to:  } \mathcal{D}(\mathbb{P}_0^n \parallel \mathbb{P}_1^n)- 2 \delta^2 \leq 0 \nn \\
&= C_{pd}^w(\delta). 
\end{align}
\end{prop}
\begin{proof}
By realizing that, for any feasible $\mathbf{Q}$, the function $\mathcal{D} \left(\mathbb{P}_0 \parallel \mathbb{P}_1(\mathbf{W}_w)\right)$ is monotonically increasing in $\mathbf{W}_w$, we have that
\begin{align}
&\min_{\mathbf{W}_w \in \mathcal{S}_w}  \mathcal{R}( \mathbf{W}_w , \mathbf{Q} ,\delta) = \mathcal{R}( \gamma_w \hat{\mathbf{I}}, \mathbf{Q} ,\delta) \nn \\
&\,\,\,\,\,\text{Subject to:  } \mathcal{D}(\mathbb{P}_0^n \parallel \mathbb{P}_1^n)- 2 \delta^2.
\end{align} 
Hence, the required result follows by using proposition \ref{prop:W_ww}.
\end{proof}

\subsection{Evaluation of $C_{pd}(\delta)$}
In light of the saddle point property established in the previous Section, in this Section we characterize 
$C_{pd}(\delta)$ by solving (\ref{eq:worst_lpd_cap}) for the optimal signaling strategy, $\mathbf{Q}_*$. We give the main result of this Section in the following theorem.
\begin{theorem}
\label{thm:main_general} 
The eigenvalue decomposition of the capacity achieving input covariance matrix that solves \eqref{eq:eq_MIMO_proplem} is given by $\mathbf{Q}_* = \mathbf{U}_b\Lambda \mathbf{U}_b^{\dagger}$ where $\mathbf{U}_b \in \mathbb{C}^{N_a \times N_a}$ is the matrix whose columns are the right singular vectors of $\mathbf{H}_b$ and $\Lambda$ is a diagonal matrix whose diagonal entries, $\Lambda_{ii}$, are given by the solution of 
\begin{align}
\label{eq:diag_opt_main}
\lambda = &(\sigma_b^2 \lambda_i^{-1}(\mathbf{W}_b) + \Lambda_{ii})^{-1} \nn \\
&+ \eta \left( \left(\dfrac{\sigma_w^2}{\gamma_w} + \Lambda_{ii}\right)^{-2}- \left(\dfrac{\sigma_w^2}{\gamma_w} + \Lambda_{ii}\right)^{-1}\right)
\end{align}
where $\lambda$ and $\eta$ are constants determined from the constraints $\mathbf{tr}\left\{\mathbf{Q}\right\} \leq P$ and (\ref{eq:metric}), respectively.
 Moreover,
 \begin{align}
 \label{eq:c_pd_cap}
 C_{pd}(\delta) &= \sum_{i=1}^{N} \log \left(1+ \dfrac{\Lambda_{ii}\lambda_i(\mathbf{W}_b)}{\sigma_b^2}\right)
 \end{align}
 where $\lambda_i$ is the $i^{th}$ non zero eigenvalue of $\mathbf{W}_b$.
\end{theorem}
\begin{proof}  
See Appendix \ref{app:main_general}.
\end{proof} 
The result of Theorem \ref{thm:main_general} provides the full characterization of $C_{pd}(\delta)$ of the considered class of channels. It can be seen that, the singular value decomposition (SVD) precoding \cite{telatar1999capacity} is the optimal signaling strategy except for the water filling strategy in (\ref{eq:diag_opt_main}) which is chosen to satisfy both power and LPD constraints. Unlike both MIMO channel without security constraint and MIMO wiretap channel, transmission with full power is, indeed, not optimal. Let
\begin{align}
 P_{th}  \triangleq \sum_i \Lambda_{ii},
\end{align}
be the maximum total power that is transmitted by Alice. An equivalent visualization of our problem is that Alice need to choose a certain power threshold, $P_{th}$, to satisfy the LPD constraint. However, again, $P_{th}$ is distributed along the \textit{eigenmodes} using conventional water filling solution. Although it is not straightforward to obtain a closed form expression\footnote{Using Mathematica, $\Lambda_{ii}$ was found to be an expression of almost 30 lines which does not provide the required insights here.} for $C_{pd}(\delta)$ and, hence, $L$, we could obtain both upper and lower bounds on $C_{pd}(\delta)$ which leads to upper and lower bounds on $L$. Based on the obtained bounds, we give the square-root law for MIMO AWGN channel in the following Theorem.
\begin{theorem}[Square-root Law of MIMO AWGN channel]
\label{thm:srl_mimo}
 For the considered class of channels, the following bounds on $C_{pd}(\delta)$ holds
\begin{align}
\sum_{i=1}^{N} \log &\left(1+ \dfrac{\sqrt{2}\sigma_w^2\delta \lambda_i(\mathbf{W}_b)}{\sigma_b^2\gamma_w \sqrt{n M}} \right) \leq C_{pd}(\delta) \nn \\
& \leq \sum_{i=1}^{N} \log \left(1+ \dfrac{\sqrt{2}\sigma_w^2 \xi \delta \lambda_i(\mathbf{W}_b)}{\sigma_b^2\gamma_w \sqrt{n M}} \right)
\end{align}
where $\xi \geq 1$ is a function of $\delta$ that approaches $1$ as $\delta$ goes to $0$. Moreover, 
\begin{align}
\label{eq:l_bounds_sk}
\sum_{i=1}^{N} \dfrac{\sigma_w^2 \lambda_i(\mathbf{W}_b)}{\sigma_b^2\gamma_w \sqrt{ M}} \leq L \leq \sum_{i=1}^{N} \dfrac{\sigma_w^2 \xi \lambda_i(\mathbf{W}_b)}{\sigma_b^2\gamma_w \sqrt{M}}.
\end{align}
\end{theorem}
\begin{proof}
We give both achievability and converse results in Appendix \ref{app:srl_mimo}.
\end{proof}
Theorem \ref{thm:srl_mimo} extends the square-root law for scalar AWGN channel to the MIMO AWGN channel. In particular, it states that Alice can transmit a maximum of $\mathcal{O}(N\sqrt{n / M})$ bits reliably to Bob in $n$ independent channel uses while keeping Willie's sum of error probabilities lower bounded by $1-\delta$. The interesting result here is that, the gain in covert rate scales linearly with number of active eigenmodes of Bob channel. Meanwhile, it scales down with the square-root of the number of active eigenmodes over Willie channel. This fact will be of great importance when we study the case of massive MIMO limit. Further, the bounds on $L$ in \eqref{eq:l_bounds_sk} can, with small effort, generate the result of Theorem 5 in \cite{wang2016fundamental} by setting $N=M=1$, $\lambda(\mathbf{W}_b)=\gamma_w$ and $\sigma_b = \sigma_w$.

It worth mentioning that, in some practical situations, compound MIMO channel can be too conservative for resource allocation. In particular, the bounded spectral norm condition in (\ref{eq:willie_compound}) not only leads us to the worst case Willie channel, but it also does not restrict its eigenvectors leaving the beamforming strategy used by Alice (SVD precoding) to be of insignificant gain in protecting against Willie. Although we believe that the eigenvectors of Willie channel plays an important role in the determination of the achievable covert rate, the ignorance of Alice about Willie channel leaves the compound framework as our best option. 

\section{Unit Rank MIMO Channel}
As pointed out in the previous Section, the distinction between the eigenvectors of Bob and Willie channels would have a considerable effect on the achievable covert rate. However, unavailability of Willie's CSI left the compound framework as the best model for $\mathbf{H}_w$. In this Section, we consider the case when either $\mathbf{H}_w$ or Both $\mathbf{H}_w$ and $\mathbf{H}_b$ are of unit rank. This scenario not only models the case when both Bob and Willie have a single antenna, but it also covers the case when they have a strong line of sight with Alice. Moreover, this scenario allow us to evaluate the effect of the eigenvectors of  $\mathbf{H}_w$ and $\mathbf{H}_b$ on the achievable covert rate. 

\subsection{Unit Rank Willie Channel}
\label{sec:main_willie_rank1}
In this Section we analyze the scenario in which only Willie channel is of unit rank. In this case, we can write $\mathbf{H}_w = \lambda_w^{1/2} v_w \mathbf{u}_w^{\dagger}$, where $v_w \in \mathbb{C}^{N_w}$ and $\mathbf{u}_w \in \mathbb{C}^{N_{a}}$ are the left and right singular vectors of $\mathbf{H}_w$. Accordingly, $\mathbf{W}_w = \lambda_w \mathbf{u}_w \mathbf{u}_w^{\dagger}$ and the product $\mathbf{W}_w\mathbf{Q}$ has only one non zero eigenvalue. The nonzero eigenvalue $\lambda(\mathbf{W}_w\mathbf{Q})$ is \textit{loosely} upper bounded by $\lambda_w \lambda_{max}(\mathbf{Q})$. Accordingly, following the same steps of the proof of Theorem \ref{thm:srl_mimo} we can get (assuming well conditioned Bob channel, i.e. $\lambda_i(\mathbf{W}_b) = \lambda_b$ for all $i$)
\begin{align}
N \log &\left(1+ \dfrac{\sqrt{2}\sigma_w^2\delta \lambda_b}{\sigma_b^2 \lambda_w \sqrt{n}} \right) \leq C_{pd}(\delta) \nn \\
& \leq N \log \left(1+ \dfrac{\sqrt{2}\sigma_w^2 \xi \delta \lambda_b}{\sigma_b^2 \lambda_w \sqrt{n}} \right)
\end{align}
consequently,
\begin{align}
\label{eq:nr_main_rank1}
N \dfrac{\sigma_w^2 \lambda_b}{\sigma_b^2 \lambda_w }   \leq L \leq N  \dfrac{\sigma_w^2 \xi \lambda_b}{\sigma_b^2 \lambda_w }.
\end{align}
Again, Bob gets $\mathcal{O}(N \sqrt{n})$ bits in $n$ independent channel uses. We note also that, the achievable covert rate increases linearly with $N$. 

\subsection{Both Channels are of Unit Rank}
\label{sec:main_bob_rank1}
Consider the case when both $\mathbf{H}_b$ and $\mathbf{H}_w$ are of unit rank. In this case, we have $N=1$. However, setting $N=1$ in the results established so far will yield a loose bounds on the achievable LPD constrained rate. The reason is that, the bound $\lambda(\mathbf{W}_e\mathbf{Q}) \leq \lambda_w\lambda_{max}(\mathbf{Q})$ is, in fact, too loose especially for large values of $N_a$. Although it is hard to establish a tighter upper bound on $\lambda(\mathbf{W}_e\mathbf{Q})$ when $\mathbf{Q}$ is of high rank, it is straightforward to obtain the exact expression for $\lambda(\mathbf{W}_e\mathbf{Q})$ when $\mathbf{Q}$ is of unit rank (which is the case when $rank\{\mathbf{H}_b\}=1$). Given that $\mathbf{H}_b = \lambda_b^{1/2} v_b \mathbf{u}_b^{\dagger}$, Alice will set $\mathbf{Q}=P_{th}\mathbf{u}_b \mathbf{u}_b^{\dagger}$. Accordingly, we have
\begin{align}
\label{eq:eig_rank1}
\lambda(\mathbf{W}_e\mathbf{Q}) &= \lambda_w P_{th} \abs{<\mathbf{u}_b,\mathbf{u}_w>}^2 \nn \\
&= \lambda_w P_{th} \cos^2(\theta),
\end{align}
where $\theta$ is the angle between $\mathbf{u}_b$ and $\mathbf{u}_w$. We give the main result of this Section in the following theorem.
\begin{theorem}
\label{thm:rank1}
If $rank\{\mathbf{H}_b\} = rank\{\mathbf{H}_e\} = 1$, then,
\begin{align}
\label{eq:C_rank1}
\min\left\{ \log \left(1+ \dfrac{\sqrt{2}\sigma_w^2\delta \lambda_b}{\sigma_b^2 \lambda_w \cos^2(\theta) \sqrt{n}} \right) , C \right\} \leq C_{pd}(\delta) \nn \\
\,\,\,\,\leq   \min\left\{\log \left(1+ \dfrac{\sqrt{2}\sigma_w^2 \xi \delta \lambda_b}{\sigma_b^2 \lambda_w \cos^2(\theta) \sqrt{n}} \right), C \right\}
\end{align}
where $C$ is the non LPD constrained capacity of Alice to Bob channel. Accordingly,
\begin{align}
\label{eq:L_rank1}
 \begin{cases}
      L= \infty, & \text{if}\ \theta=\pi/2 \\
      \dfrac{\sigma_w^2 \lambda_b}{\sigma_b^2 \lambda_w \cos^2(\theta)}  \leq L \leq   \dfrac{\sigma_w^2 \xi \lambda_b}{\sigma_b^2 \lambda_w \cos^2(\theta)}, & \text{otherwise}
    \end{cases}.
\end{align}
\end{theorem}
\begin{proof}
Follows directly by substituting (\ref{eq:eig_rank1}) into (\ref{eq:D_gen_cont}) and following the same steps as in the proof of Theorem \ref{thm:srl_mimo} while realizing that $P_{th} \leq P$.
\end{proof}
Theorem \ref{thm:rank1} proves that Alice can transmit a maximum of $\mathcal{O}(\sqrt{n}/\cos^2(\theta))$ bits reliably to Bob in $n$ independent channel uses while keeping Willie's sum of error probabilities lower bounded by $1-\delta$. In the statement of Theorem \ref{thm:rank1}, the minimum is taken since the first term diverges as $\theta \rightarrow \pi/2$, i.e., when $\mathbf{u}_b$ and $\mathbf{u}_w$ are orthogonal. In such case, we will have $L=\infty$. This fact proves that, over MIMO channel Alice can communicate at full rate to Bob without being detected by Willie. An interesting question is, how rare is the case of having $\mathbf{u}_b$ and $\mathbf{u}_w$ to be orthogonal? When the angle of the vectors (i.e., the antenna orientation) are chosen uniformly at random, as the number of antennas at Alice gets large, we will see in Section \ref{sec:massive} that  $\cos^2(\theta)$ approaches $0$ exponentially fast with the number of antennas at Alice.

\begin{remark}
It should not be inferred from the results in this Section that the unit rank Bob channel can offer covert rate better than that of higher rank. In fact, we used a loose upper bound on the eigenvalue of Willie's channel for higher rank case. That is due to the technical difficulty in setting tight bounds to the power received by Willie when Bob channel has higher rank. Also, we see that unit rank channel offers better covert rate than that shown under the compound settings for Willie channel. That is because, unlike the scenario of this Section, compound settings does not restrict the eigenvectors of Willie's channel.   
\end{remark} 
\section{$C_{pd}(\delta)$ Without Shared Secret}
\label{sec:PD_nosk}
So far, we have established fundamental limits of covert communication over MIMO AWGN channel under the assumption that the codebook generated by Alice is kept secret from Willie (or at least a secret of sufficiet length). In this Section, we study the LPD communication problem without this assumption. The assumption of keeping the codebook generated by Alice secret from Willie (or at least a secret of sufficient length \cite{bash2013limits,bloch2016covert}) is common in all studies of covert communication. Without this assumption, LPD condition cannot be met along with arbitrarily low probability of error at Bob, since, when Willie is informed about the codebook, he can decode the message using the same decoding strategy as that of Bob \cite{bash2013limits}. Also we note that, acheiving covertness does not require positive secrecy rate over the underlying wiretap channel. Indeed, recalling the expression for relative entropy at Willie
\begin{align}
\label{eq:D_gen_cont_nosk}
\mathcal{D} \left(\mathbb{P}_0 \parallel \mathbb{P}_1 \right) =&  \underbrace{\log \abs{\dfrac{1}{ \sigma_w^2 } \mathbf{H}_w \mathbf{Q}\mathbf{H}_w^{\dagger} + \mathbf{I}_{N_w}}}_{\text{Willie's Channel Capacity}}  \nonumber  \\
&+ \underbrace{\mathbf{tr}\left\{\left[\dfrac{1}{\sigma_w^2 } \mathbf{H}_w \mathbf{Q}\mathbf{H}_w^{\dagger} + \mathbf{I}_{N_w}\right]^{-1}\right\} - N_w}_{\leq 0 \text{, Willie's penalty due to codebook ignorance}},      
\end{align}
we observe that, the first term in (\ref{eq:D_gen_cont_nosk}) is the channel capacity of Willie channel with the implicit assumption of the knowledge of the codebook generated by Alice. In particular, the first term in (\ref{eq:D_gen_cont_nosk}) equals to $\mathcal{I}(\mathbf{x};(\mathbf{z},\mathbf{H}_w))$. Meanwhile, it can be easily verified that the remaining difference term in (\ref{eq:D_gen_cont_nosk}) is always non positive. This term represents Willie's penalty from his ignorance of the codebook. Analogous result for the scalar AWGN channel can be found in \cite{bash2013limits} for which the same arguments can be made. This in fact provides an interpretation to the scenario on which the secrecy capacity of the main channel may be zero, meanwhile, Alice still can \textit{covertly} communicate to Bob. Let us define
\begin{align}
\label{eq:L_hat}
\hat{L} \triangleq \lim_{\epsilon \downarrow 0} \varliminf_{n \rightarrow \infty} \dfrac{K_n(\delta,\epsilon)}{n\sqrt{2 \delta^2}}.
\end{align} 
Following Proposition \ref{prop:l_mimo} and theorem \ref{thm:l_mimo}, we can show that
\begin{align}
\hat{L} = \lim_{n  \rightarrow \infty} \dfrac{n}{\sqrt{2\delta^2}} C_{pd}(\delta).
\end{align} 
Observe that, in the definition of $\hat{L}$, we used the normalization over $n$ instead of $\sqrt{n}$. Now suppose that Alice chooses $\mathbf{Q}$ such that the first term in (\ref{eq:D_gen_cont_nosk}), which is the capacity of Willie' channel, is upper bounded by $2\delta^2/n$. Indeed this signaling strategy satisfies the LPD metric \eqref{eq:metric} and, thus, achieves covertness. Moreover, we have $\lim_{n \rightarrow \infty}\mathcal{I}(\mathbf{x};(\mathbf{z},\mathbf{H}_w)) = 0$, thus, strong secrecy condition is also met. In particular, if $\mathcal{I}(\mathbf{x};(\mathbf{z},\mathbf{H}_w)) \leq 2\delta^2/n$, Willie can reliably decode at most $2\delta^2$ nats of Alice's message in $n$ independent channel uses. It worth mentioning that, requiring $\lim_{n \rightarrow \infty}\mathcal{I}(\mathbf{x};(\mathbf{z},\mathbf{H}_w)) = 0$ is more restrictive than the strong secrecy condition. In principle, if Alice has a message $\mathbf{m}$ to transmit, strong secrecy condition requires $\lim_{n \rightarrow \infty}\mathcal{I}(\mathbf{m};(\mathbf{z},\mathbf{H}_w)) = 0$. Meanwhile, since $\mathbf{m}= f^{-1}(\mathbf{x})$ for some encoding function $f: \mathbf{m} \mapsto \mathbf{x}$, we have $\mathcal{I}(\mathbf{m};(\mathbf{z},\mathbf{H}_w)) \leq \mathcal{I}(\mathbf{x};(\mathbf{z},\mathbf{H}_w))$. Now, $C_{pd}(\delta)$ without any shared secret can be reformulated as follows:
\begin{align}
\label{eq:eq_MIMO_NOSK}
C_{pd}(\delta) = \max_{\substack{\mathbf{Q} \succeq \mathbf{0} \\ \mathbf{tr}(\mathbf{Q}) \leq P}}  &\log \abs{\mathbf{I}_{N_a} + \dfrac{\mathbf{W}_b\mathbf{Q}}{\sigma_b^2}}  \\
\label{eq:metric_nosk}
\text{Subject to:  } &\log \abs{\mathbf{I}_{N_a} + \dfrac{\mathbf{W}_w\mathbf{Q}}{\sigma_w^2}} - 2 \delta^2/n \leq 0.
\end{align} 
In light of the saddle point property established in Section \ref{sec:main_informed}, we characterize 
$C_{pd}(\delta)$ without shared secret by solving (\ref{eq:eq_MIMO_NOSK}) for the optimal signaling strategy, $\mathbf{Q}_*$. We give the main result of this Section in the following theorem.
\begin{theorem}
\label{thm:main_nosk}
The eigenvalue decomposition of the capacity achieving input covariance matrix that solves \eqref{eq:eq_MIMO_NOSK} is given by $\mathbf{Q}_* = \mathbf{u}_b\Lambda \mathbf{u}_b^{\dagger}$ where $\mathbf{u}_b \in \mathbb{C}^{N_a \times N_a}$ is the matrix whose columns are the right singular vectors of $\mathbf{H}_b$ and $\Lambda$ is a diagonal matrix whose diagonal entries, $\Lambda_{ii}$, are given by the solution of 
\begin{align}
\label{eq:diag_opt_main_nosk}
\lambda = &(\sigma_b^2 \lambda_i^{-1}(\mathbf{W}_b) + \Lambda_{ii})^{-1}- \eta  \left(\dfrac{\sigma_w^2}{\gamma_w} + \Lambda_{ii}\right)^{-1}
\end{align}
where $\lambda$ and $\eta$ are constants determined from the constraints $\mathbf{tr}\left\{\mathbf{Q}\right\} \leq P$ and (\ref{eq:metric}), respectively.
Moreover,
 \begin{align}
 \label{eq:c_pd_cap_nosk}
 C_{pd}(\delta) &= \sum_{i=1}^{N} \log \left(1+ \dfrac{\Lambda_{ii}\lambda_i(\mathbf{W}_b)}{\sigma_b^2}\right)
 \end{align}
 where $\lambda_i$ is the $i^{th}$ non zero eigenvalue of $\mathbf{W}_b$.
\end{theorem}
\begin{proof}  
See Appendix \ref{app:main_nosk}.
\end{proof} 
Theorem \ref{thm:main_nosk} provides the full characterization of the $C_{pd}(\delta)$ of the considered class of channels without requiring any shared secret between Alice and Bob. Again, it is not straightforward to obtain a closed form expression for $C_{pd}(\delta)$. Thus, we obtain both upper and lower bounds on $C_{pd}(\delta)$ as we did in Section \ref{sec:main_informed}. Based on the obtained bounds, we give the square-root law for MIMO AWGN channel without shared secret in the following theorem.
\begin{theorem}
\label{thm:srl_mimo_nosk}
 For the considered class of channels without any shared secret between Alice and Bob, the following bounds on $C_{pd}(\delta)$ holds
\begin{align}
\sum_{i=1}^{N} \log &\left(1+ \dfrac{2\sigma_w^2 \delta^2 \lambda_i(\mathbf{W}_b)}{\sigma_b^2\gamma_w n M} \right) \leq C_{pd}(\delta) \nn \\
& \leq \sum_{i=1}^{N} \log \left(1+ \dfrac{2\sigma_w^2 \xi \delta^2 \lambda_i(\mathbf{W}_b)}{\sigma_b^2\gamma_w n M} \right)
\end{align}
where $\xi = \dfrac{nM }{n M - 2 \delta^2 }$. Accordingly,
\begin{align}
\sum_{i=1}^{N} \dfrac{\sqrt{2}\sigma_w^2 \delta \lambda_i(\mathbf{W}_b)}{\sigma_b^2\gamma_w M} \leq \hat{L}  \leq  \sum_{i=1}^{N} \dfrac{\sqrt{2}\sigma_w^2 \xi \delta \lambda_i(\mathbf{W}_b)}{\sigma_b^2\gamma_w M} 
\end{align} 
\end{theorem}
\begin{proof}
We give both achievability and converse results in Appendix \ref{app:srl_mimo_nosk}.
\end{proof}
Theorem \ref{thm:srl_mimo_nosk} extends the result of Theorem \ref{thm:srl_mimo} to the scenario when Alice and Bob do not share any form of secret. It proves that Alice can transmit a maximum of $\mathcal{O}(N/M)$ bits reliably to Bob in $n$ independent channel uses while keeping Willie's sum of error probabilities lower bounded by $1-\delta$.

Now let us consider the case when both $\mathbf{H}_b$ and $\mathbf{H}_w$ are of unit rank. Under this assumption, we give converse and achievability results of $C_{pd}(\delta)$ over Alice to Bob channel without a shared secret in the following theorem.
\begin{theorem}
\label{thm:rank1_nosk}
If $rank\{\mathbf{H}_b\} = rank\{\mathbf{H}_e\} = 1$, then, $\delta$-PD constrained capacity over Alice to Bob channel without a shared secret between Alice and Bob is bounded as
\begin{align}
\label{eq:rank1_nosk_thm}
\min\left\{ \log \left(1+ \dfrac{2\sigma_w^2 \delta^2 \lambda_b}{\sigma_b^2 \lambda_w \cos^2(\theta) n} \right) , C \right\} \leq C_{pd}(\delta) \nn \\
\,\,\,\,\leq   \min\left\{\log \left(1+ \dfrac{2 \sigma_w^2 \xi \delta^2 \lambda_b}{\sigma_b^2 \lambda_w \cos^2(\theta) n} \right), C \right\}
\end{align}
where $\xi$ is as defined in Theorem \ref{thm:srl_mimo_nosk} and $C$ is the non LPD constrained capacity of Alice to Bob channel. Accordingly,
\begin{align}
    \begin{cases}
      \hat{L} = \infty, & \text{if}\ \theta= \pi/2 \\
      \dfrac{\sqrt{2}\sigma_w^2 \delta \lambda_b}{\sigma_b^2 \lambda_w \cos^2(\theta)} \leq \hat{L} \leq   \dfrac{\sqrt{2}\sigma_w^2 \xi \delta \lambda_b}{\sigma_b^2 \lambda_w \cos^2(\theta)}, & \text{otherwise}
    \end{cases}
\end{align}
\end{theorem}
\begin{proof}
Follows directly by substituting (\ref{eq:eig_rank1}) into \eqref{eq:metric_nosk} and following the same steps as in the proof of Theorem \ref{thm:srl_mimo_nosk} while realizing that $P_{th} \leq P$.
\end{proof}
Again, in \eqref{eq:rank1_nosk_thm}, the minimum is taken since the first term diverges as $\theta \rightarrow \pi/2$, i.e., when $\mathbf{u}_b$ and $\mathbf{u}_w$ are orthogonal. The theorem proves that Alice can transmit a maximum of $\mathcal{O}(1/\cos^2(\theta))$ bits reliably to Bob in $n$ independent channel uses while keeping Willie's sum of error probabilities lower bounded by $1-\delta$. This fact proves that, over MIMO channel Alice can communicate at full rate to Bob without being detected by Willie without requiring Alice and Bob to have any form of shared secret.
\section{Covert Communication with Massive MIMO}
\label{sec:massive}
In Theorems \ref{thm:rank1} and \ref{thm:rank1_nosk}, it was shown that Alice can communicate at full rate with Bob without being detected by Willie whenever $\cos(\theta) = 0$ regardless of the presence of a shared secret. In this Section, we study the behavior of covert rate as the number of antennas scale, which we call the massive MIMO limit, with and without codebook availability at Willie. In particular, the high beamforming capability of the massive MIMO system can provide substantial gain in the achievable LPD rate. However, a quantitative relation between the achievable LPD rate and the number of transmitting antennas seems to be unavailable. To that end, we address the question: how does the achievable LPD rate scale with the number of transmitting antennas? We also study how does the presence of a shared secret between Alice and Bob affects the scaling of the covert rate in the massive MIMO limit. Before we answer these questions, we state some necessary basic results on the inner product of unit vectors in higher dimensions \cite{cai2013distributions}.
\subsection{Basic Foundation}
 In this Section, we reproduce some established results on the inner product of unit vectors in higher dimensions.
\begin{lemma} 
\label{lema:lema1}
[Proposition 1 in \cite{cai2013distributions}] Let $\mathbf{a}$ and $\mathbf{b}$ any two vectors in the unit sphere in $\mathbb{C}^{p}$ chosen uniformly at random. Let $\theta = \cos^{-1}( <a,b>)$ be the angle between them. Then
\begin{align}
Pr\left(\abs{\theta - \dfrac{\pi}{2}} \leq \zeta \right) \geq 1- K \sqrt{p} (\cos \zeta)^{p-2}
\end{align}
for all $p \geq 2$ and $\zeta \in \left(0,\dfrac{\pi}{2}\right)$ where $K$ is a universal constant.
\end{lemma}
The statement of Lemma \ref{lema:lema1} states that, the probability that any two vectors chosen uniformly at random being orthogonal increases exponentially fast with  the dimension $p$. Indeed, note that, for any $0 < a < 1$, $a^p$ has the same decay rate as $(2-a)^{-p}$. Thus, the probability that $\theta$ is within $\zeta$ from $\pi/2$ scales like $(2-\cos(\zeta))^{p-2}/K\sqrt{p}$. 

\begin{cor} 
\label{cor:cor1}
Let $\mathbf{a}$ and $\mathbf{b}$ any two vectors in the unit sphere in $\mathbb{C}^{p}$ chosen uniformly at random and let $\theta = \cos^{-1}( <a,b>)$ be the angle between them. Let $A,B \in \mathbb{C}^{p \times p}$ be two matrices of unit rank generated as $A = \lambda_a aa^{\dagger}$ and $B= \lambda_b bb^{\dagger}$. Then, the probability that the eigenvalue of of the product $\lambda(AB)$ approaches $0$ grows to $1$ exponentially fast with the dimension $p$.
\end{cor}
\begin{proof}
It can be easily verified that $\lambda(AB) = \lambda_a \lambda_b \cos^2(\theta)$. Using Lemma \ref{lema:lema1}, we see that, the probability that $\theta$ approaches $\pi/2$ increases exponentially with $p$. Hence, the probability that $\cos(\theta)$ approaches $0$ increases in the same order. Then so is $\cos^2(\theta)$.    
\end{proof} 
 \subsection{Massive MIMO Limit With Shared Secret}
 In the previous Section it was demonstrated that, in higher dimensions every two independent vectors chosen uniformly at random are orthogonal with very high probability. More generally, using spherical invariance \cite{cai2013distributions}, given $\mathbf{u}_b$, for any $\mathbf{u}_w$ chosen uniformly at random in $\mathbb{C}^{N_a}$, the result of Lemma \ref{lema:lema1} still holds. This scenario typically models the scenario when Alice knows her channel to Bob, meanwhile, she models $\mathbf{u}_w$ as a uniform random unit vector.
 
 Recall that, when Alice has $nN_a$ bits to transmit, two alternative options are available for her. Either she splits the incoming stream into $N_a$ streams of $n$ bits each and use each stream to select one from $2^{n}$ messages for each single antenna, or, use the entire $n N_a$ bits to choose from $2^{n N_a}$ message. The latter of these alternatives provides a gain factor of $N_a$ in the error exponent, of course, in the expense of much greater complexity \cite{gallager1968information,telatar1999capacity}. However, in the restrictive LPD scenario, Alice would choose the latter alternative as to achieve the best decoding performance at Bob. Therefore, we need to analyze the LPD rate in the limiting case of the product $nN_a$ when both $n$ and $N_a$ grow. Of course, the number of antennas at Alice is a physical resource which can not be compared to $n$ that can approach $\infty$ very fast. The more interesting question is, how fast $\cos^2(\theta)$ approaches $0$ as $N_a$ increase. As illustrated in corollary \ref{cor:cor1}, we know that $\cos(\theta)$ approaches $0$ \textit{exponentially fast} with $N_a$. Consequently, we conclude that $\cos^2(\theta)$ , also, approaches $0$ \textit{exponentially fast} with $N_a$. For proper handling of the scaling of $K_n(\delta,\epsilon)$ in massive MIMO limit, let us define
\begin{align}
\label{eq:S}
S \triangleq \lim_{\epsilon \downarrow 0} \varliminf_{nN_a \rightarrow \infty} \dfrac{K_n(\delta,\epsilon)}{N_a\sqrt{2 n  \delta^2}}.
\end{align} 
Note that the in the definition of $S$, both $n$ and $N_a$ are allowed to grow without bound compared to $L$ in which only $n$ was allowed to grow while $N_a$ was treated as constant. Now observe that, following Proposition \ref{prop:l_mimo} and Theorem \ref{thm:l_mimo}, we can show that
\begin{align}
S = \lim_{nN_a \rightarrow \infty} N_a \sqrt{\dfrac{n }{2\delta^2}} C_{pd}(\delta).
\end{align}
We give the result of the massive MIMO limit with a pre-shared secret between Alice and Bob in the following Theorem.
\begin{theorem}
\label{thm:limit_main}
Assume that $rank\{\mathbf{H}_b\} = rank\{\mathbf{H}_e\} = 1$. Given $\mathbf{u}_b$, for any $\mathbf{u}_w$ chosen uniformly at random, $C_{pd}(\delta)$ is as given in Theorem \ref{thm:rank1} and
\begin{align}
S = \infty.
\end{align}
Moreover, $K_n$ grows like $\sqrt{\dfrac{n}{K^2 N_a}}(1+\dfrac{c}{\sqrt{n}})^{(N_a-2)/2}$ where $K$ is a universal constant and $c =\left(\dfrac{\sqrt{2} \sigma_w^2 \delta}{\lambda_w P}\right)$.
\end{theorem}
\begin{proof}
Combining the result of Theorem \ref{thm:rank1} and Corollary \ref{cor:cor1}, multiplying \eqref{eq:C_rank1} by $N_a \sqrt{\dfrac{n }{2\delta^2}}$ and taking the limit as both of $n$ and $N_a$ tend to infinity we obtain 
\begin{align}
\label{eq:s_infty1}
\lim_{N_a \rightarrow \infty} \lim_{n \rightarrow \infty} N_a &\sqrt{\dfrac{n }{2\delta^2}} C_{pd}(\delta) \nn \\
  &= \lim_{N_a \rightarrow \infty} N_a \dfrac{\sigma_w^2 \lambda_b}{\sigma_b^2 \lambda_w \cos^2(\theta)} \nn \\
  &= \infty,
\end{align}
where the last equality follow since $\cos^2(\theta) \rightarrow 0$ as $N_a \rightarrow \infty$. On the other hand, we also can verify that
\begin{align}
\label{eq:s_infty2}
\lim_{n \rightarrow \infty} \lim_{N_a \rightarrow \infty} N_a &\sqrt{\dfrac{n }{2\delta^2}} C_{pd}(\delta) = \infty.
\end{align}
To show how $K_n$ scales in this massive MIMO limit, we first note that, for fixed $N_a$, $K_n$ scales like $\sqrt{n}$. Also note that, $S=\infty$ implies that LPD constraint becomes inactive and full non-LPD capacity is achieved. This happens when the quantity $C_{pd}(\delta) = C$. The question we adress now is, how does $C_{pd}(\delta)$ behave in between these two extreme regimes. Following the same steps of the proof of Theorem \ref{thm:srl_mimo}, we can obtain the following bound on $P_{th}$:
\begin{align}
P_{th} \leq \min\left\{\dfrac{\sqrt{2} \sigma_w^2 \delta}{\sqrt{n}\lambda_w \cos^2(\theta)},P \right\}.
\end{align}
Thus, we have $P_{th} = P$, and hence $C_{pd}(\delta) = C$, when
\begin{align}
P \leq \dfrac{\sqrt{2} \sigma_w^2 \delta}{\sqrt{n} \lambda_w \cos^2(\theta)}
\end{align}
equivalently,
\begin{align}
\label{eq:thm_limit}
&\cos^2(\theta) \leq \dfrac{\sqrt{2} \sigma_w^2 \delta}{\sqrt{n} \lambda_w P} \nn \\
\Rightarrow &\abs{\theta - \pi/2} \leq \pi/2 - \cos^{-1} \left(\sqrt{\dfrac{\sqrt{2} \sigma_w^2 \delta}{\sqrt{n} \lambda_w P}}\right).
\end{align}
This happens with probability no less than:
\begin{align}
\label{eq:pr_cap}
Pr( C_{pd}(\delta) = C )  \geq 1-K\sqrt{N_a} \left(1-\dfrac{\sqrt{2} \sigma_w^2 \delta}{\sqrt{n}\lambda_w P}\right)^{(N_a -2)/2}
\end{align}
where \eqref{eq:pr_cap} follows by setting $\zeta$ in Lemma \ref{lema:lema1} equal to the RHS of (\ref{eq:thm_limit}) and using the following basic trigonometry facts: $\cos(\pi/2 - x) = \sin(x)$ and $\sin(\cos^{-1}(x)) = \sqrt{1-x^2}$. It can be seen that, the probability that $C_{pd}(\delta) = C$ scales as $(1+g)^{({N_a - 2})/2}/ K \sqrt{N_a}$ up to $1$, where $g=\left(\dfrac{\sqrt{2} \sigma_w^2 \delta}{\sqrt{n}\lambda_w P}\right)$.  
\end{proof}

Theorem \ref{thm:limit_main} states that Alice can communicate at full rate to Bob while satisfying the LPD constraint \eqref{eq:metric}. Note that, the limit in both orders yields $S=\infty$.

 As $N_a \rightarrow \infty$, the radiation pattern of a wireless MIMO transmitter becomes so extremely directive (pencil beam). We call this limit the \textit{wired limit of wireless MIMO communication}. In the wired limit, Willie cannot detect Alice's transmission unless he wiretaps this \textit{virtual wire}. Theorem \ref{thm:limit_main} provides a rigorous characterization of the wired limit of wireless MIMO communication. In principle, it answers the fundamental question: How fast does the LPD constrained rate increase with the number of antennas at Alice? It can be seen that the probability that Alice fully utilizes the channel scales like $2^{(N_a-2)/2} / K \sqrt{N_an}$ up to $1$ using the same justification given after Lemma \ref{lema:lema1}.
\subsection{Massive MIMO Limit Without Shared Secret}
\label{sec:massive_nosk}
In Section \ref{sec:PD_nosk} it was shown that, only diminishing covert rate, $\mathcal{O}(N/M)$, can be achieved without requiring a shared secret between Alice and Bob. Again, we note that this diminishing rate was shown to be achievable when Willie's channel is isotropic. Also, we have shown shown that, in the massive MIMO limit, the achievable covert rate grows exponentially with the number of transmitting antennas when there is a shared secret between Alice and Bob. Thus, it is also instructive to consider LPD communication problem without a shared secret in the massive MIMO limit. As illustrated in Section \ref{sec:example}, if Alice has CSI of both channels, not only can she communicate covertly and reliably at full rate whenever the eigen directions of both channels are orthogonal, but also she does not need a shared secret to achieve this rate. Building on our analysis in Section \ref{sec:massive}, we give the massive MIMO limit of the $\delta$-PD capacity when there is no shared secret between Alice and Bob.

Now, let us consider the scenario in which $\mathbf{u}_w$ is chosen uniformly at random and fixed once chosen. For proper handling of the scaling of $K_n(\delta,\epsilon)$ in massive MIMO limit without a shred secret, let us define
\begin{align}
\label{eq:S_nosk}
\hat{S} \triangleq \lim_{\epsilon \downarrow 0} \varliminf_{nN_a \rightarrow \infty} \dfrac{K_n(\delta,\epsilon)}{n N_a\sqrt{2 \delta^2}}.
\end{align} 
Observe that, unlike $S$, $K_n(\delta,\epsilon)$ is normalized to $n$ instead of $\sqrt{n}$ in the expression of $\hat{S}$. Now, following Proposition \ref{prop:l_mimo} and Theorem \ref{thm:l_mimo}, we can show that
\begin{align}
\hat{S} = \lim_{nN_a \rightarrow \infty}  \dfrac{n N_a }{\sqrt{2\delta^2}} C_{pd}(\delta).
\end{align}
We give the result of this scenario in the following Theorem.
\begin{theorem}
 \label{thm:limit_main_nosk}
Assume that $rank\{\mathbf{H}_b\} = rank\{\mathbf{H}_e\} = 1$ and suppose that there is no shared secret between Alice and Bob. Given $\mathbf{u}_b$, for any $\mathbf{u}_w$ chosen uniformly at random, $C_{pd}(\delta)$ is as given in Theorem \ref{thm:rank1} and
\begin{align}
\label{eq:lim_cap_nosk}
\hat{S} = \infty.
\end{align}
Moreover, $K_n$ grows like $\sqrt{\dfrac{1}{K^2 N_a}}(1+\dfrac{c}{n})^{(N_a-2)/2}$ where $K$ is a universal constant and $c =\left(\dfrac{\sqrt{2} \sigma_w^2 \delta}{\lambda_w P}\right)$.
\end{theorem} 
\begin{proof} 
The proof follows exactly the same steps as in the proof of Theorem \ref{thm:limit_main}.    
\end{proof}
Again, it can be seen that Alice achieve the maximum achievable non-LPD rate even under the LPD constraint.However, the rate at which $C_{pd}(\delta)$ converges to $C$ is much slower, compared to the case with shared secret codebook. Hence, it can be deduced from Theorems \ref{thm:rank1_nosk} and \ref{thm:limit_main_nosk} that, in the limit of large $N_a$, Alice can transmit $\mathcal{O}(n)$ bits in $n$ independent channel uses while satisfying the LPD constraint without the need for any form of shared secret. Even though, it has to be considered that the number of antennas required at Alice under this scenario is much larger than that when she shares a secret of sufficient length with Bob. The following numerical example demonstrates the covert rates in massive MIMO limit with and without a shared secret between Alice and Bob.

 \begin{example}
 Assume that Alice intend to use the channel for $n=10^{9}$ times over a channel of bandwidth of $10MHz$, hence, $\sqrt{n} = 3.1623\times 10^{4}$. Suppose that Alice is targeting $\delta = 10^{-2}$. Let $\sigma_w^2= \sigma_b^2  = 10^{-2}$ and $\lambda_w = \lambda_b = 10^{-3}$. Assume that Alice is targeting $SNR=15dB$ at Bob, hence, $P=316.228$. Then, for $N_a = 100$ it can be verified that, Alice can transmit $\mathcal{O}(n)$ covert bits instead of $\mathcal{O}(\sqrt{n})$. Observe that, Alice needed only $N_a = 100$ to communicate \textit{covertly} at near full rate to Bob. Also note that, at $6GHz$, two dimensional array of $100$ elements can fit within an area of a single sheet of paper. See Fig. (\ref{fig:LPD_CAP_NOSK}) for the relation between the $\delta$-PD capacity and number of transmitting antenna for different values of number of antennas at Willie with $\delta = 10^{-2}$. 
 \end{example}
 \begin{figure}[ht]
\centering
\includegraphics[width = 3.4 in, height = 2.3 in]{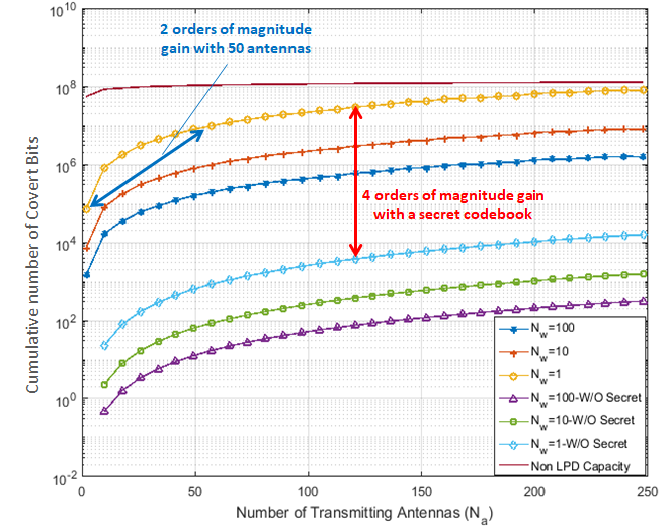}
\caption{The relation between achievable covert rate with and without a shared secret (plotted in log scale), in bits per second, and $N_a$ for different values of $N_w$ with target $\delta=10^{-2}$. It shows that Alice can communicate near full rate with $N_a $ around $100$ when she share a secret with Bob. A large gap can be observed when there is no shared secret.  
\label{fig:LPD_CAP_NOSK}}
\end{figure}

As can be seen from Fig. \eqref{fig:LPD_CAP_NOSK}, Alice could achieve a covert rate very close to the non-LPD constrained capacity of her channel to Bob with $N_a \geq 100$ with a block of length $n=10^9$. Also we see that there is a significant gap (nearly 4 orders of magnitude) between the achievable covert rate with and without a preshared secret. Although both rates converges to $C$ as $N_a \rightarrow \infty$, we see that without a shared secret, the number of antennas required to achieve near full rate is significantly greater than that required when Alice and Bob are sharing a secret of sufficient length. For practical consideration, this result leaves the massive MIMO limit of the $\delta$-PD capacity without a shared secret of theoretical interest only.


\section{Discussion}
\textbf{Impact of CSI level.} Throughout this paper we have considered perfect CSI at Alice about the her channel to Bob. In practical scenarios, this assumption might not hold true as CSI always suffer from imperfection due to e.g. channel estimation error or non error free CSI feedback link. Despite these potential impairments, we argue that the case of imperfect CSI at Alice does not affect the obtained results. That is because we have made the assumption that Willie has perfect CSI as well about his channel. Further, in case when Alice has absolutely no CSI is of special interest as communication under these critical conditions may not allow Bob to share his CSI to Alice, specially when dealing with passive Bob. In this scenario, we can verify that the Alice can transmit $\mathcal{O}(N/M\sqrt{n})$ bits in $n$ independent channel uses (details are omitted here). We see that, the covert rate scales with $N/M$ compared to $N/\sqrt{M}$ when Alice has CSI. More interestingly, in massive MIMO limit with absolutely no CSI at Alice, we can verify that the covert rate $\rightarrow 0$ even with $N_a \rightarrow \infty$. While it was recognized as the most favorable scenario for Alice when she has CSI, massive MIMO may make matters worse when she has absolutely no CSI. Contrary, when Alice has CSI about both channels, covert rates up to the non LPD constrained rate of the channel may be achieved (under certain conditions, see Section \ref{sec:example}) without the need for a shared secret. As it is the case for MIMO channel with no secrecy constraint, CSI availability play an important role in achieving higher covert rates.

\textbf{Impact of Willie's ignorance.} All results obtained in this paper assumes that Willie has perfect CSI about his channel and he is aware of his channel noise statistics. Ignorance of Willie about one of these parameters is expected to have positive impact on the achievable covert rate. For example, in \cite{reliable_deniable}, it was shown that $\mathbf{O}(n)$ bits can be transmitted reliably with low probability of detection over BSC whose error probability is unknown to Willie except that it is drawn from a known interval. This scenario is subject for future research. 

\textbf{Length of the Shared Secret.} In our analysis, we have considered scenarios in which the entire codebook is either available or unavailable at Willie. Meanwhile, secrets of shorter length were reported to be enough for fulfilling LPD requirements. For scalar AWGN channel, it was shown that the required secret secret length is in order of $\mathcal{O}(\log n \sqrt{n})$ \cite{bash2013limits}. Similar result was established for DMC in \cite{wang2016fundamental}. In this work, while we showed that there is a significant gap between the achievable covert rate with and without a shared secret, the minimum length of the required shared secret has not been addressed. 

\section{Summary and Conclusions}
We have established the limits of LPD communication over the MIMO AWGN channel. In particular, using relative entropy as our LPD metric, we studied the maximum codebook size, $K_n(\delta,\epsilon)$, for which Alice can guarantee reliability and LPD conditions are met. We first showed that, the optimal codebook generating input distribution under $\delta$-PD constraint is the zero-mean Gaussian distribution. We based our arguments on the the principle of minimum relative entropy. For an isotropic Willie channel, we showed that Alice can transmit $\mathcal{O}(N \sqrt{n/M})$ bits reliably in $n$ independent channel uses, where $N$ and $M$ are the number of active eigenmodes of Bob and Willie channels, respectively. Further, we evaluated the scaling rates of $K_n(\delta,\epsilon)$ in the limiting regimes for the number of channel uses (asymptotic block length) and the number of antennas (massive MIMO). We showed that, while the square-root law still holds for the MIMO-AWGN channel, the number of bits that can be transmitted covertly scales exponentially with the number of transmitting antennas. More precisely, for a unit rank MIMO channel, we show that $K_n(\delta,\epsilon)$ scales as $\sqrt{\dfrac{n}{K^2 N_a}}(1+\dfrac{c}{\sqrt{n}})^{(N_a-2)/2}$ where $N_a$ is the number of transmitting antennas, $K$ is a universal constant and $c$ is constant independent on $n$ and $N_a$. Also, we derived the scaling of $K_n(\delta,\epsilon)$ \textit{with no shared secret} between Alice and Bob. In particular, we showed that achieving better covert rate is a resource arm race between Alice, Bob and Willie: Alice can transmit $\mathcal{O}(N/M)$ bits reliably in $n$ independent channel uses, i.e., the covert rate is in the order of the ratio between active eigenmodes of both channels. Despite this diminishing rate, in the massive MIMO limit, Alice can still achieve higher covert rate up to the non LPD constrained capacity of her channel to Bob, yet, with a significantly greater number of antennas. Although the covert rates both with and without a shared secret are shown to converge the non LPD constrained capacity as $N_a \rightarrow \infty$, numerical evaluations showed that without a shared secret, the number of antennas required to achieve near full rate can be orders of magnitude greater. The practical implication of our result is that, MIMO has the potential to provide a substantial increase in the file sizes that can be covertly communicated subject to a reasonably low delay.

\bibliographystyle{IEEEtran}
\bibliography{IEEEabrv,References}

\begin{appendices}
\section{Kullback–Leibler Divergence at Willie}
\label{app:KL}
Assuming that Willie is informed about its own channel to Alice, Willie's observation when Alice is silent is distributed as $\mathbb{P}_0$, meanwhile, it takes the distribution $\mathbb{P}_1$ whenever Alice is active where
\begin{align}
\mathbb{P}_0 = \abs{\pi \Sigma_0}^{-1} \exp\left(-\mathbf{z}^{\dagger} \Sigma_0^{-1} \mathbf{z}\right), \nonumber \\
\mathbb{P}_1 = \abs{\pi \Sigma_1}^{-1} \exp\left(-\mathbf{z}^{\dagger} \Sigma_1^{-1} \mathbf{z}\right),
\end{align} 
where $\Sigma_0 = \sigma_w^2 \mathbf{I}_{N_w}$ and $\Sigma_1 =  \mathbf{H}_w\mathbf{Q}\mathbf{H}_w^{\dagger} + \sigma_w^2 \mathbf{I}_{N_w}$. Here, $\mathbf{Q}= \mathbb{E}\left[\mathbf{x}\mathbf{x}^{\dagger}\right]$ is the covariance matrix of signal transmitted by Alice. Note that, the choice of $\mathbf{Q}$ is highly dependent on the amount of CSI available at Alice. Thus, we evaluate the KL divergence at Willie in general, then, in the following Sections we will discuss the effect of CSI availability at Alice on the performance of Willie's optimal detector. Assuming that Willie channel is known and fixed, the KL divergence between $\mathbb{P}_0$ and $\mathbb{P}_1$ is given as follows:
\allowdisplaybreaks[2]
\begin{eqnarray}
\label{eq:D_gen}
\mathcal{D}  &=& \mathbb{E}_{\mathbb{P}_0} \left[\log \mathbb{P}_0 - \log \mathbb{P}_1\right] \nonumber \\
&=&  \mathbb{E}_{\mathbb{P}_0} \left[-\log \abs{\pi \Sigma_0} - \mathbf{z}^{\dagger} \Sigma_0^{-1} \mathbf{z}\right. \nn \\
 &&\;\;\;\;\;\;\;+ \left. \log \abs{\pi \Sigma_1} + \mathbf{z}^{\dagger} \Sigma_1^{-1} \mathbf{z} \right] \nonumber \\
&=&  \log \dfrac{\abs{\Sigma_1}}{\abs{\Sigma_0}} +\mathbb{E}_{\mathbb{P}_0} \left[ \mathbf{z}^{\dagger} \Sigma_1^{-1} \mathbf{z}  - \mathbf{z}^{\dagger} \Sigma_0^{-1} \mathbf{z} \right] \nonumber \\
&=&  \log \abs{\Sigma_1\Sigma_0^{-1}} + \mathbb{E}_{\mathbb{P}_0} \left[ \mathbf{z}^{\dagger} \left(\Sigma_1^{-1} -\Sigma_0^{-1} \right)  \mathbf{z} \right] \nonumber \\
&=&  \log \abs{\Sigma_1\Sigma_0^{-1}} + \mathbb{E}_{\mathbb{P}_0} \left[  \mathbf{tr}\left\{\left(\Sigma_1^{-1} -\Sigma_0^{-1} \right) \mathbf{z}\mathbf{z}^{\dagger}\right\} \right] \nonumber \\
&=&  \log \abs{\Sigma_1\Sigma_0^{-1}} + \mathbf{tr}\left\{\Sigma_1^{-1} \Sigma_0\right\} - \mathbf{tr}\left\{\Sigma_0^{-1} \Sigma_0\right\} \nonumber \\
&=&  \log \abs{\Sigma_1\Sigma_0^{-1}} + \mathbf{tr}\left\{\Sigma_1^{-1} \Sigma_0\right\} - N_w
\end{eqnarray}
Now observe that, $\Sigma_1 \Sigma_0^{-1} = \dfrac{1}{\sigma_w^2}\mathbf{H}_w\mathbf{Q}\mathbf{H}_w^{\dagger} + \mathbf{I}_{N_w}$ and that $\abs{\dfrac{1}{\sigma_w^2}\mathbf{H}_w\mathbf{Q}\mathbf{H}_w^{\dagger} + \mathbf{I}_{N_w}} = \abs{\dfrac{1}{\sigma_w^2}\mathbf{W}_w\mathbf{Q} + \mathbf{I}_{N_a}}$. Define $\mathbf{W}_w \triangleq \mathbf{H}_w^{\dagger}\mathbf{H}_w$ we note that, the non-zero eigenvalues of $\mathbf{H}_w\mathbf{Q}\mathbf{H}_w^{\dagger}$ and $\mathbf{W}_w\mathbf{Q}$ are identical, hence, we ca write:
\allowdisplaybreaks[2]
\begin{eqnarray}
\label{eq:D_gen_cont_proof}
\mathcal{D} \left(\mathbb{P}_0 \parallel \mathbb{P}_1 \right) &=&  \underbrace{\log \abs{\dfrac{1}{ \sigma_w^2 } \mathbf{H}_w \mathbf{Q}\mathbf{H}_w^{\dagger} + \mathbf{I}_{N_w}}}_{\text{Willie's Channel Capacity}}  \nonumber  \\
&&+ \underbrace{\mathbf{tr}\left\{\left[\dfrac{1}{\sigma_w^2 } \mathbf{H}_w \mathbf{Q}\mathbf{H}_w^{\dagger} + \mathbf{I}_{N_w}\right]^{-1}\right\} - N_w}_{\leq 0 \text{, Willie's penalty due to codebook ignorance}}        \nonumber  \\
&=&  \log \prod_{i=1}^{N_w} \left(1+ \dfrac{ \lambda_i({\mathbf{W}}_w \mathbf{Q})}{\sigma_w^2}\right) \nonumber \\
&&+ \sum_{i=1}^{N_w} \left(1+ \dfrac{ \lambda_i({\mathbf{W}}_w \mathbf{Q})}{\sigma_w^2}\right)^{-1} - N_w \nonumber \\
&=& \sum_{i=1}^{N_w} \big[ \log \left(1+ \dfrac{ \lambda_i({\mathbf{W}}_w \mathbf{Q})}{\sigma_w^2}\right) \nonumber \\
&&+ \left(1+ \dfrac{ \lambda_i({\mathbf{W}}_w \mathbf{Q})}{\sigma_w^2}\right)^{-1} - 1 \big]\nonumber \\
&=& \sum_{i=1}^{N_w} \big[ \log \left(1+\dfrac{ \lambda_i({\mathbf{W}}_w \mathbf{Q})}{\sigma_w^2}\right) \nn \\
&& - \left(1+ \left(\dfrac{ \lambda_i({\mathbf{W}}_w \mathbf{Q})}{\sigma_w^2}\right)^{-1}\right)^{-1}\big],
\end{eqnarray} 
where $\lambda_{i}({\mathbf{W}}_w \mathbf{Q})$ is the $i^{th}$ eigenvalue of ${\mathbf{W}}_w \mathbf{Q}$. 
\section{Achievability Proof of Proposition \ref{prop:l_mimo}}
\label{app:l_mimo_prop}
We show the achievability for Gaussian input. As discussed in \cite{wang2016fundamental}, the sequence $\left\{K_n \right\}$ is achievable provided that:
\begin{align}
\label{eq:varlimsup}
\varlimsup_{n \rightarrow \infty} \dfrac{K_n}{\sqrt{n}} \leq P-\liminf_{n \rightarrow \infty} \dfrac{1}{\sqrt{n}} \log \dfrac{f^{\times n}(\mathbf{z}^n | \mathbf{x}^n)}{f^{\times n}(\mathbf{z}^n)},
\end{align}
where $P-\liminf$ denotes the limit inferior in probability, namely, the largest number such that the probability that the random variable in consideration is greater than this number tends to one as $n$ tends to infinity. Meanwhile, $f^{\times n} (\mathbf{x}^n)$ denotes the $n^{th}$ extension of the probability density function of a random vector $\mathbf{x}$. Thus, we need to show
\begin{align}
\dfrac{1}{\sqrt{n}} \log \dfrac{f^{\times n}(\mathbf{z}^n | \mathbf{x}^n)}{f^{\times n}(\mathbf{z}^n)} - \sqrt{n} \mathcal{I}(f_n(\mathbf{x}),f_n(\mathbf{z}|\mathbf{x})) \rightarrow 0 
\end{align}
in probability as $n$ tends to infinity. First observe that,
 \begin{align}
f^{\times n}(\mathbf{z}^n | \mathbf{x}^n) &= \prod_{i=1}^n \abs{\pi \Sigma_0}^{-1} \exp\left(-(\mathbf{z}_i -\mathbf{x}_i )^{\dagger} \Sigma_0^{-1} (\mathbf{z}_i - \mathbf{x}_i)\right), \nonumber \\
f^{\times n}(\mathbf{z}^n ) &=  \prod_{i=1}^n \abs{\pi \Sigma_1}^{-1} \exp\left(-\mathbf{z}_i^{\dagger} \Sigma_1^{-1} \mathbf{z}_i\right),
\end{align} 
where $\Sigma_0 = \sigma_w^2 \mathbf{I}_{N_w}$, $\Sigma_1 =  \mathbf{H}_w\mathbf{Q}\mathbf{H}_w^{\dagger} + \sigma_w^2 \mathbf{I}_{N_w}$ and let $\mathbf{Q}= \mathbb{E}\left[\mathbf{x}\mathbf{x}^{\dagger}\right]$ be chosen such that \eqref{eq:metric} is satisfied. Note that, $\mathbf{Q}$ has to be a decreasing function of $n$. Then,
\begin{align}
\dfrac{f^{\times n}(\mathbf{z}^n | \mathbf{x}^n)}{f^{\times n}(\mathbf{z}^n )} &= \abs{\Sigma_1 \Sigma_0^-1}^n \times \nn \\ 
&\exp\left(\sum_{i=1}^n \mathbf{tr}(\Sigma_1^{-1} \mathbf{z}_i \mathbf{z}_i^{\dagger}) - \sum_{i=1}^n \mathbf{tr}(\Sigma_0^{-1} \mathbf{e}_i \mathbf{e}_i^{\dagger})\right).
\end{align}
Accordingly,
\begin{align}
\dfrac{1}{\sqrt{n}} \log \dfrac{f^{\times n}(\mathbf{z}^n | \mathbf{x}^n)}{f^{\times n}(\mathbf{z}^n )} &= \sqrt{n} \log \abs{\Sigma_1 \Sigma_0^-1} + \nn \\ 
&\dfrac{1}{\sqrt{n}} \left(\sum_{i=1}^n \mathbf{tr}(\Sigma_1^{-1} \mathbf{z}_i \mathbf{z}_i^{\dagger}) - \sum_{i=1}^n \mathbf{tr}(\Sigma_0^{-1} \mathbf{e}_i \mathbf{e}_i^{\dagger})\right),
\end{align}
whose expectation can be found as
\begin{align}
\mathbb{E}\left[\dfrac{1}{\sqrt{n}} \log \dfrac{f^{\times n}(\mathbf{z}^n | \mathbf{x}^n)}{f^{\times n}(\mathbf{z}^n )}\right] &= \sqrt{n} \log \abs{\Sigma_1 \Sigma_0^-1} + \nn \\ 
&\dfrac{1}{\sqrt{n}} \left(\sum_{i=1}^n \mathbf{tr}(\Sigma_1^{-1} \Sigma_1) - \sum_{i=1}^n \mathbf{tr}(\Sigma_0^{-1} \Sigma_0)\right) \nn \\
&= \sqrt{n} \log \abs{\Sigma_1 \Sigma_0^-1} \nn \\
&= \sqrt{n} \log \abs{\mathbf{I}_{N_w} + \dfrac{1}{\sigma_w^2}\mathbf{H}_w \mathbf{Q} \mathbf{H}_w^{\dagger} } \nn \\
&=  \sqrt{n} \mathcal{I}(f_n(\mathbf{x}),f_n(\mathbf{z}|\mathbf{x})).
\end{align}
It then follows by Chebyshev’s inequality that, for any constant $a >0$,
\begin{align}
Pr\left(\abs{\dfrac{1}{\sqrt{n}} \log \dfrac{f^{\times n}(\mathbf{z}^n | \mathbf{x}^n)}{f^{\times n}(\mathbf{z}^n)} - \sqrt{n} \mathcal{I}(f_n(\mathbf{x}),f_n(\mathbf{z}|\mathbf{x}))} \geq a\right) \nn \\
\leq \dfrac{1}{a^2} \mathbf{var}\left(\dfrac{1}{\sqrt{n}} \log \dfrac{f^{\times n}(\mathbf{z}^n | \mathbf{x}^n)}{f^{\times n}(\mathbf{z}^n)}\right)
\end{align}
and, it remains to show that
\begin{align}
\lim_{n \rightarrow \infty} \mathbf{var} \left(\dfrac{1}{\sqrt{n}} \log \dfrac{f^{\times n}(\mathbf{z}^n | \mathbf{x}^n)}{f^{\times n}(\mathbf{z}^n )}\right) = \mathbf{0}.
\end{align} 
Note that,
\begin{align}
\label{eq:covar}
\mathbf{var} &\left(\dfrac{1}{\sqrt{n}} \log \dfrac{f^{\times n}(\mathbf{z}^n | \mathbf{x}^n)}{f^{\times n}(\mathbf{z}^n )}\right) \nn \\
&= \mathbf{var}\left( \dfrac{1}{\sqrt{n}} \left(\sum_{i=1}^n \mathbf{z}_i^{\dagger} \Sigma_1^{-1} \mathbf{z}_i  - \sum_{i=1}^n \mathbf{e}_i^{\dagger} \Sigma_0^{-1} \mathbf{e}_i \right)\right) \nn \\
&= \dfrac{1}{n} \sum_{i=1}^n \mathbf{var}\left(   \mathbf{z}_i^{\dagger} \Sigma_1^{-1} \mathbf{z}_i  - \mathbf{e}_i^{\dagger} \Sigma_0^{-1} \mathbf{e}_i \right) \nn \\
&= \mathbf{var}\left(   \mathbf{z}_i^{\dagger} \Sigma_1^{-1} \mathbf{z}_i  - \mathbf{e}_i^{\dagger} \Sigma_0^{-1} \mathbf{e}_i \right) \nn \\
&= \mathbb{E} \left[ \left(   \mathbf{z}_i^{\dagger} \Sigma_1^{-1} \mathbf{z}_i  - \mathbf{e}_i^{\dagger} \Sigma_0^{-1} \mathbf{e}_i \right) \right. \nn \\
&\;\;\;\;\;\;\;\;\;\left. \left(   \mathbf{z}_i^{\dagger} \Sigma_1^{-1} \mathbf{z}_i  - \mathbf{e}_i^{\dagger} \Sigma_0^{-1} \mathbf{e}_i \right)^{\dagger} \right] \nn \\
&= \mathbb{E} \left[ \mathbf{z}_i^{\dagger} \Sigma_1^{-1} \mathbf{z}_i  \mathbf{z}_i^{\dagger} \Sigma_1^{-1} \mathbf{z}_i \right. \nn \\
&\;\;\;\;\;\;\;\;\; - \mathbf{z}_i^{\dagger} \Sigma_1^{-1} \mathbf{z}_i \mathbf{e}_i^{\dagger} \Sigma_0^{-1} \mathbf{e}_i  \nn \\
&\;\;\;\;\;\;\;\;\; - \mathbf{e}_i^{\dagger} \Sigma_0^{-1} \mathbf{e}_i \mathbf{z}_i^{\dagger} \Sigma_1^{-1} \mathbf{z}_i  \nn \\
&\;\;\;\;\;\;\;\;\; \left. - \mathbf{e}_i^{\dagger} \Sigma_0^{-1} \mathbf{e}_i \mathbf{e}_i^{\dagger} \Sigma_0^{-1} \mathbf{e}_i \right] \nn \\
\end{align}
Now observe that, since $\mathbf{Q} \rightarrow 0$ as $n$ tends to infinity, we can verify that each term in \eqref{eq:covar} tends to zero as $n$ tends to infinity. $\blacksquare$
\section{Proof of Theorem \ref{thm:l_mimo}}
\label{app:l_mimo}
We show that the limit in \eqref{eq:l_mimo} always exists. Note that, for every $n$,  $f_n(\mathbf{x})$ is zero mean Gaussian. Let $\mathbf{Q}_n^* = \mathbb{E}_n\left[\mathbf{x}\mathbf{x}^{\dagger}\right]$ be
\begin{align}
\mathbf{Q}_n^* =  \argmax_{\substack{\mathbf{Q} \succeq\mathbf{0} \\ \mathbf{tr}\left(\mathbf{Q}\right) \leq P}} \mathcal{I}(f_n(\mathbf{x}),f_n(\mathbf{y}))
\end{align}
where the maximum is subject to \eqref{eq:metric}. Hence, we have 
\begin{align}
\max_{\substack{\mathbf{Q} \succeq\mathbf{0} \\ \mathbf{tr}\left(\mathbf{Q}\right) \leq P}} \mathcal{I}(f_n(\mathbf{x}),f_n(\mathbf{y})) = \log \abs{\mathbf{I}+\dfrac{\mathbf{H}_b \mathbf{Q}_n^* \mathbf{H}_b^{\dagger}}{\sigma_b^2}}.
\end{align}
Now, we have two cases to consider:
\begin{enumerate}
\item $\mathcal{D}(\mathbb{P}_0^n \parallel \mathbb{P}_1^n) = 0$. In this case, $\delta$ can be made $0$ causing the limit to be infinity.
\item$\mathcal{D}(\mathbb{P}_0^n \parallel \mathbb{P}_1^n) > 0$. In this case, $\mathbf{Q}_n^*$ has to be a decreasing function of $n$, otherwise, the constraint \eqref{eq:metric} can not be met. In this case, the limit is $\geq 0$ and $< \infty$.
\end{enumerate}
In either cases, the limit exist and, hence, limit can be used in place of limit inferior and, also, the order of limit and maximum can be interchanged.

\section{Proof of Proposition \ref{prop:W_ww}}
\label{app:W_ww}
It is enough to show that
\begin{align}
\mathbf{W}_w^*  &= \argmax_{\mathbf{W}_w \in \mathcal{S}_w}\mathcal{D} \left(\mathbb{P}_0 \parallel \mathbb{P}_1(\mathbf{W}_w)\right) \nn\\
&= \gamma_w \hat{\mathbf{I}},
\end{align}
hence, we need to show that the function $\mathcal{D} \left(\mathbb{P}_0 \parallel \mathbb{P}_1(\mathbf{W}_w)\right)$ is monotonically increasing in $\mathbf{W}_w$, i.e., $\mathcal{D}_1 \triangleq \mathcal{D} \left(\mathbb{P}_0 \parallel \mathbb{P}_1(\mathbf{W}_{w1})\right) \geq \mathcal{D} \left(\mathbb{P}_0 \parallel \mathbb{P}_1(\mathbf{W}_{w2})\right) \triangleq \mathcal{D}_2 $ whenever $\mathbf{W}_{w1} \succcurlyeq \mathbf{W}_{w2}$. Recalling the expression of  $\mathcal{D} \left(\mathbb{P}_0 \parallel \mathbb{P}_1(\mathbf{W}_w)\right)$ in (\ref{eq:D_gen_cont}), we note that the function $\log \abs{\mathbf{I}+\mathbf{W}\mathbf{Q}}$ is monotonically increasing in $\mathbf{W}$ for any $\mathbf{Q}$. Meanwhile, the second term in (\ref{eq:D_gen_cont}) is negative and decreases monotonically in $\mathbf{W}$. Even though, we have that
\allowdisplaybreaks[2]
\begin{eqnarray}
\label{eq:D_worst}
\mathcal{D}_1  - \mathcal{D}_2 &=& \mathbb{E}_{\mathbb{P}_0} \left[\log \mathbb{P}_0 - \log \mathbb{P}_1(\mathbf{W}_{w1}) \right. \nn \\ 
&\quad& \left. \,\,\,\,\,\,\,\,\,\,\,\, - \log \mathbb{P}_0 + \log \mathbb{P}_1(\mathbf{W}_{w2})\right] \nonumber \\
&=& \mathbb{E}_{\mathbb{P}_0} \left[\log \mathbb{P}_1(\mathbf{W}_{w2}) - \log \mathbb{P}_1(\mathbf{W}_{w1})\right] \nonumber \\
&\stackrel{(a)}{\geq}& - \mathbb{E}_{\mathbb{P}_0}  \left[ \log \dfrac{\mathbb{P}_1(\mathbf{W}_{w1})}{ \mathbb{P}_1(\mathbf{W}_{w2})} \right] \nonumber \\
&\stackrel{(b)}{\geq}& 0,
\end{eqnarray}
where $(a)$ follows from Jensen's inequality using convexity of $\log$ and $(b)$ follows because, with some standard matrix algebra as in Appendix \ref{app:KL}, we can show that, for $\mathbf{W}_{w1} \succcurlyeq \mathbf{W}_{w2}$, we have $\mathbb{E}_{\mathbb{P}_0} \left[\log \dfrac{\mathbb{P}_1(\mathbf{W}_{w1})}{ \mathbb{P}_1(\mathbf{W}_{w2})} \right] \leq 1$. $\blacksquare$


\section{Proof of Theorem \ref{thm:main_general}}
\label{app:main_general}
We solve (\ref{eq:eq_MIMO_proplem}) for $\mathbf{W}_w=\mathbf{W}_w^*$. Without loss of generality, assume that $N_w \geq N_a$, hence, $\mathbf{W}_w^* = \gamma_w\mathbf{I}_{N_w}$. For $N_w < N_a$, we can set the $\gamma_w = 0$ for the $N_a -N_w$ minimum eigenvalues of $\mathbf{Q}_*$. Plugging $\mathbf{W}_w^*$ into (\ref{eq:D_gen_cont}), we get \footnote{For simplicity, we write $\mathcal{D}^n$ instead of $\mathcal{D} \left(\mathbb{P}_0^n \parallel \mathbb{P}_1^n \right)$}
\allowdisplaybreaks[2]
\begin{align}
\label{eq:D_main_informed}
\mathcal{D}^n &=& n\left(\log \abs{\mathbf{I}_{N_a}+ \dfrac{\gamma_w}{\sigma_w^2} \mathbf{Q}} + \mathbf{tr}\left\{\left[\mathbf{I}_{N_a}+ \dfrac{\gamma_w}{\sigma_w^2} \mathbf{Q}\right]^{-1}\right\} - N_w\right)
\end{align}

Now observe that,
\begin{align}
\abs{\mathbf{I}+\mathbf{W}_b \mathbf{Q}} \leq \prod_{i=1}^{N_a} (1+\lambda_i(\mathbf{W}_b)\lambda_i(\mathbf{W}_b))
\end{align} 
with equality if and only if $\mathbf{W}_b$ and $\mathbf{Q}$ have the same eigenvectors and note that this choice does not affect (\ref{eq:D_main_informed}) since $\mathbf{W}_w^*$ is isotropic. Hence, the eigenvectors of $\mathbf{Q}_*$ is the same as the eigenvectors of $\mathbf{W}_b$ which is the same as the right singular vectors of $\mathbf{H}_b$.

Now, we form the following Lagrange dual problem
\begin{align}
\label{eq:lagrange}
\mathcal{L}= & \log  \abs{\mathbf{I}_{N_a} +  \dfrac{1}{\sigma_b^2} \mathbf{W}_b \mathbf{Q}} + \lambda (\mathbf{tr}(\mathbf{Q}) -P) - \mathbf{tr}(\mathbf{M}\mathbf{Q})\nonumber \\
&+ \eta \left[\mathcal{D} - \dfrac{2\delta^2}{n}\right],
\end{align}
where $\lambda,\; \eta \geq 0$ are the Lagrange multipliers that penalize violating the power and LPD constraints, respectively, and $\mathbf{M} \succeq \mathbf{0}$ penalizes the violation of the constraint $\mathbf{Q} \succeq \mathbf{0}$. Where the associated KKT conditions can be expressed as:
\begin{align}
\label{eq:KKT}
&\lambda,\;\eta \geq 0 , \;\;\; \lambda (\mathbf{tr}(\mathbf{Q}) -P) = 0 ,\;\;\; \mathbf{M}\mathbf{Q} = \mathbf{0}, \nonumber \\
& \mathbf{Q} \succeq \mathbf{0}, \;\;\; \mathbf{M} \succeq \mathbf{0},\;\;\; \mathbf{tr}(\mathbf{Q}) \leq P,
\end{align}
where the equality constraints in (\ref{eq:KKT}) ar the complementary slackness conditions. Note that, (\ref{eq:eq_MIMO_proplem}) is not a concave problem in general. Thus, KKT conditions are not sufficient for optimality. Yet, since the constraint set is compact and convex and the objective function is continuous, KKT conditions are necessary for optimality. Hence, we proceed by finding the stationary points of the gradient of the dual Lagrange problem in the direction of $\mathbf{Q}$ and obtain the stationary points that solve the KKT conditions. By inspecting the objective function at these points, the global optimum can be identified. To identify the stationary points of the Lagrangian (\ref{eq:lagrange}), we get its gradient with respect to $\mathbf{Q}$ as follows:

\begin{align}
\bigtriangledown_{\mathbf{Q}}\mathcal{L} = &-\left[\mathbf{I}_{N_a} + \dfrac{1}{\sigma_b^2} \mathbf{W}_b \mathbf{Q}  \right]^{-1}\dfrac{\mathbf{W}_b}{\sigma_b^2}  + \lambda \mathbf{I}_{N_a} -\mathbf{M} \nn \\
&+ \dfrac{ \eta \gamma_w}{\sigma_b^2} \left(\left[\mathbf{I}_{N_a} + \dfrac{\gamma_w}{\sigma_w^2} \mathbf{Q}  \right]^{-1} -  \left[\mathbf{I}_{N_a} + \dfrac{\gamma_w}{\sigma_w^2} \mathbf{Q}  \right]^{-2}\right)\nn \\
= &-\left[\mathbf{I}_{N_a} + \dfrac{1}{\sigma_b^2} \mathbf{W}_b \mathbf{Q}  \right]^{-1}\dfrac{\mathbf{W}_b}{\sigma_b^2}  + \lambda \mathbf{I}_{N_a} -\mathbf{M}  \nn \\
&+ \eta  \left( \left[\dfrac{\sigma_w^2}{\gamma_w}\mathbf{I}_{N_a} +  \mathbf{Q}  \right]^{-1}  - \left[\dfrac{\sigma_w^2}{\gamma_w}\mathbf{I}_{N_a} +  \mathbf{Q}  \right]^{-2}\right)
\end{align}
Assume, without loss of generality, that $\mathbf{Q} \succ \mathbf{0}$, then, from $\mathbf{M}\mathbf{Q} = \mathbf{0}$ it follows that $\mathbf{M} = \mathbf{0}$. Now, from $\bigtriangledown_{\mathbf{Q}}\mathcal{L} = 0$ we obtain
\begin{align}
\lambda \mathbf{I}_{N_a} = &\left[\sigma_b^2 \mathbf{W}_b^{-1}  + \mathbf{Q}  \right]^{-1} \nn \\
&+ \eta  \left( \left[\dfrac{\sigma_w^2}{\gamma_w}\mathbf{I}_{N_a} +  \mathbf{Q}  \right]^{-2}  - \left[\dfrac{\sigma_w^2}{\gamma_w}\mathbf{I}_{N_a} +  \mathbf{Q}  \right]^{-1}\right)
\end{align}
Since we know that $\mathbf{Q}_*$ and $\mathbf{W}_b$ have the same eigenvectors, hence, the eigenvalues of $\mathbf{Q}_*$, $\Lambda_{ii}$, can be found from
\begin{align}
\lambda = &(\sigma_b^2 \lambda_i^{-1}(\mathbf{W}_b) + \Lambda_{ii})^{-1} \nn \\
&+ \eta \left( \left(\dfrac{\sigma_w^2}{\gamma_w} + \Lambda_{ii}\right)^{-2}- \left(\dfrac{\sigma_w^2}{\gamma_w} + \Lambda_{ii}\right)^{-1}\right),
\end{align}
as required. $\blacksquare$
\section{Proof of Theorem \ref{thm:srl_mimo}}
\label{app:srl_mimo}
\textbf{Achievability.} Starting from \eqref{eq:D_main_informed}, we obtain
\begin{eqnarray}
\label{eq:D_acheiv}
\mathcal{D} &\stackrel{(a)}{\leq}& \mathbf{tr}\left\{ \dfrac{\gamma_w}{\sigma_w^2} \mathbf{Q}\right\} + \mathbf{tr}\left\{\left[\mathbf{I}_{N_a}+ \dfrac{\gamma_w}{\sigma_w^2} \mathbf{Q}\right]^{-1}\right\} - N_w \nn \\
&\stackrel{(b)}{=}& \sum_{i=1}^{N_w} \left(\dfrac{\gamma_w \Lambda_{ii}}{\sigma_w^2} + \dfrac{1}{1+  \dfrac{\gamma_w \Lambda_{ii}}{\sigma_w^2}} -1\right) \nn \\
&=&  \sum_{i=1}^{N_w} \left(\dfrac{\gamma_w \Lambda_{ii}}{\sigma_w^2} - \dfrac{\dfrac{\gamma_w \Lambda_{ii}}{\sigma_w^2}}{1+  \dfrac{\gamma_w \Lambda_{ii}}{\sigma_w^2}}\right)  \nn \\
&\stackrel{(c)}{\leq}&  N_w\left( \dfrac{\gamma_w P_{th}}{N\sigma_w^2} - \dfrac{\dfrac{\gamma_w P_{th}}{N\sigma_w^2}}{1+  \dfrac{\gamma_w P_{th}}{N\sigma_w^2}} \right) \nn \\
&=&  N_w\left(\dfrac{\dfrac{\gamma_w^2 P_{th}^2}{N^2\sigma_w^4}}{1+  \dfrac{\gamma_w P_{th}}{N\sigma_w^2}} \right), \nn \\
\end{eqnarray} 
where $(a)$ follows from the inequality $\log \abs{A} \leq \mathbf{tr}\left\{A - \mathbf{I}\right\}$, $(b)$ is straightforward matrix algebra, $(c)$ follows since the RHS is maximized for $\Lambda_{ii} = P_{th}/N$ for all $i$. For Alice to ensure \eqref{eq:metric}, we need the RHS of (\ref{eq:D_acheiv}$(e)$) to be less than or equal $2\delta^2$. After some manipulation, if Alice sets
\begin{align}
\label{eq:p_th_acheiv}
P_{th} \leq \dfrac{\sqrt{2}N\sigma_w^2\delta}{\gamma_w\sqrt{n N_w}},
\end{align} 
we can verify that \eqref{eq:metric} is satisfied. Now let Alice set $P_{th}$ for (\ref{eq:p_th_acheiv}) to be met with equality. Given that choice of $P_{th}$, the LPD constraint is met and the rest of the problem is that of choosing the input distribution to maximize the achievable rate. The solution of the problem is then the SVD precoding with conventional water filling \cite{telatar1999capacity,tse2005fundamentals} as follows:
\begin{align}
\label{eq:diag_q_general}
\Lambda_{ii} = \left\{
	\begin{array}{ll}
		   (\mu - \sigma_b^2 \lambda_i^{-1}(\mathbf{W}_b))^+& \mbox{for } 1 \leq i \leq N \\
		0& \mbox{for } N < i \leq N_a, \
	\end{array}
 \right.        
\end{align} 
where $N = \min\{N_a,N_b\}$, $\lambda_i$ is the $i^{th}$ non zero eigenvalue of $\mathbf{W}_b$ and $x^+=\max\{0,x\}$. Further,  $\mu$ is a constant chosen to satisfy the power constraint $\mathbf{tr}\{\Lambda\} = P_{th}$. Accordingly, the following rate is achievable over Alice to Bob channel:
\begin{align}
\label{eq:cap_thm1}
R_{pd}(\delta) &= \sum_{i=1}^{N} \log \left(1+ \dfrac{(\mu - \sigma_b^2 \lambda_i^{-1}(\mathbf{W}_b))^+\lambda_i(\mathbf{W}_b)}{\sigma_b^2}\right) \nn \\
&= \sum_{i=1}^{N} \left(\log \left(\dfrac{\mu \lambda_i(\mathbf{W}_b)}{\sigma_b^2}\right)\right)^+.
\end{align}

However, it is technically difficult to expand \eqref{eq:cap_thm1} to check the applicability of the square-root law. Therefore, we obtain an achievable rate assuming that Alice splits $P_{th}$ equally across active eigenmodes of her channel to Bob. Note that, this rate is indeed achievable since it is less than or equal to \eqref{eq:cap_thm1}. Also note that, when Alice to Bob channel is well conditioned, the power allocation in \eqref{eq:diag_q_general} turns into equal power allocation. Then, the following rate is achievable:
\begin{align}
\label{eq:R_acheiv}
R(\delta) = \sum_{i=1}^{N} \log \left(1+ \dfrac{\sqrt{2}\sigma_w^2\delta \lambda_i(\mathbf{W}_b)}{\sigma_b^2\gamma_w\sqrt{n N_w}} \right).
\end{align}
 Now suppose that Alice uses a code of rate $\hat{R} \leq R(\delta)$, then, Bob can obtain
 \begin{align}
 n \hat{R} &\leq n  \sum_{i=1}^{N} \log \left(1+ \dfrac{\sqrt{2}\sigma_w^2\delta \lambda_i(\mathbf{W}_b)}{\sigma_b^2\gamma_w \sqrt{n N_w}} \right) \nn \\
 &\stackrel{(a)}{\leq}   \sum_{i=1}^{N} \sqrt{n/N_w} \dfrac{\sqrt{2}\sigma_w^2\delta \lambda_i(\mathbf{W}_b)}{\sigma_b^2\gamma_w \ln 2}
 \end{align}
bits in $n$ independent channel uses, where $(a)$ follows from $\ln(1+x) \leq x$, and note that the inequality is met with equality for sufficiently large $n$. Now, assume that $\lambda_i(\mathbf{W}_b) = \lambda_b$ for all $1 \leq i \leq N$, i.e., Bob's channel is well conditioned. Then,
Bob can obtain
 \begin{align}
 n R(\delta) \leq N \sqrt{n/N_w} \dfrac{\sqrt{2}\sigma_w^2\delta \lambda_b}{\sigma_b^2\gamma_w \ln 2}
 \end{align}
bits in $n$ independent channel uses since the inequality is met with equality for sufficiently large $n$. $\blacksquare$

\textbf{Converse.} To show the converse, we assume the most favorable scenario for Alice to Bob channel when $\mathbf{H}_b$ is well conditioned. That is because, the rate Alice can achieve over a well conditioned channel to Bob sets an upper bound to that can be achieved over any other channel of the same Frobenius norm \cite{tse2005fundamentals}. Then, Alice will split her power equally across active eigenmodes of her channel to Bob. Now, Let us choose $\xi \geq 1$ such that,
\begin{align}
\label{eq:log_det_ineq}
\log \abs{\mathbf{I}_{N_a}+ \dfrac{\gamma_w}{\sigma_w^2} \mathbf{Q}} - N_w \geq  \mathbf{tr}\left\{ \dfrac{\gamma_w}{\sigma_w^2} \mathbf{Q}\right\} - \xi N_w
\end{align}
and note that for small $\mathcal{D}$, $\xi$ is a function of $\delta$ that approaches $1$ as $\delta \rightarrow 0$. Hence, combining \eqref{eq:D_main_informed} with \eqref{eq:log_det_ineq} we obtain:
\begin{eqnarray}
\label{eq:D_converse}
\mathcal{D} &\stackrel{(a)}{\geq}& \mathbf{tr}\left\{ \dfrac{\gamma_w}{\sigma_w^2} \mathbf{Q}\right\} + \mathbf{tr}\left\{\left[\mathbf{I}_{N_a}+ \dfrac{\gamma_w}{\sigma_w^2} \mathbf{Q}\right]^{-1}\right\} - \xi N_w \nn \\
&\stackrel{(b)}{=}& N_w \left(\dfrac{\gamma_w P_{th}}{N\sigma_w^2} + \dfrac{1}{1+  \dfrac{\gamma_w P_{th}}{N\sigma_w^2}} - \xi\right) \nn \\
\end{eqnarray} 
following the same steps as in the achievability proof, we can insure that
\begin{align}
\label{eq:p_th_converse}
P_{th} \leq \dfrac{\sqrt{2}N \xi \sigma_w^2\delta}{\gamma_w\sqrt{n N_w}},
\end{align} 
otherwise, Alice can not ensure that \eqref{eq:metric} is satisfied. Now let Alice set $P_{th}$ equals to the RHS of (\ref{eq:p_th_converse}). Then, we can verify that:
\begin{align}
\label{eq:R_aconverse}
C_{pd}(\delta) \leq \sum_{i=1}^{N} \log \left(1+ \dfrac{\sqrt{2} \xi \sigma_w^2\delta \lambda_b}{\sigma_b^2\gamma_w\sqrt{n  N_w}} \right).
\end{align}
 Now suppose that Alice uses a code of rate $ R(\delta) \leq C_{pd}(\delta)$, then, Bob can obtain
 \begin{align}
 n R(\delta) &\leq n  N \log \left(1+ \dfrac{\sqrt{2}\xi\sigma_w^2\delta \lambda_b}{\sigma_b^2\gamma_w \sqrt{n N_w}} \right) \nn \\
 &\leq   N \sqrt{n/N_w} \dfrac{\sqrt{2}\xi\sigma_w^2\delta \lambda_b}{\sigma_b^2\gamma_w \ln 2}
 \end{align}
bits in $n$ independent channel uses since the inequality is met with equality for sufficiently large $n$. $\blacksquare$
\section{Proof of Theorem \ref{thm:main_nosk}}
\label{app:main_nosk}
We solve (\ref{eq:eq_MIMO_NOSK}) for $\mathbf{W}_w=\mathbf{W}_w^*$. Without loss of generality, assume that $N_w \geq N_a$, hence, $\mathbf{W}_w^* = \gamma_w\mathbf{I}_{N_w}$. For $N_w < N_a$, we can set the $\gamma_w = 0$ for the $N_a -N_w$ minimum eigenvalues of $\mathbf{Q}_*$. 
Now observe that,
\begin{align}
\abs{\mathbf{I}+\mathbf{W}_b \mathbf{Q}} \leq \prod_{i=1}^{N_a} (1+\lambda_i(\mathbf{W}_b)\lambda_i(\mathbf{W}_b))
\end{align} 
with equality if and only if $\mathbf{W}_b$ and $\mathbf{Q}$ have the same eigenvectors and note that this choice does not affect (\ref{eq:D_main_informed}) since $\mathbf{W}_w^*$ is isotropic. Hence, the eigenvectors of $\mathbf{Q}_*$ is the same as the eigenvectors of $\mathbf{W}_b$ which is the same as the right singular vectors of $\mathbf{H}_b$.

Now, we form the following Lagrange dual problem
\begin{align}
\label{eq:lagrange_nosk}
\mathcal{L}= & \log  \abs{\mathbf{I}_{N_a} +  \dfrac{1}{\sigma_b^2} \mathbf{W}_b \mathbf{Q}} + \lambda (\mathbf{tr}(\mathbf{Q}) -P) - \mathbf{tr}(\mathbf{M}\mathbf{Q})\nonumber \\
&+ \eta \left[\log  \abs{\mathbf{I}_{N_a} +  \dfrac{1}{\sigma_w^2} \mathbf{W}_w \mathbf{Q}} - 2\delta^2/n\right],
\end{align}
where $\lambda,\; \eta \geq 0$ are the Lagrange multipliers that penalize violating the power and LPD constraints, respectively, and $\mathbf{M} \succeq \mathbf{0}$ penalizes the violation of the constraint $\mathbf{Q} \succeq \mathbf{0}$. Where the associated KKT conditions can be expressed as:
\begin{align}
\label{eq:KKT_nosk}
&\lambda,\;\eta \geq 0 , \;\;\; \lambda (\mathbf{tr}(\mathbf{Q}) -P) = 0 ,\;\;\; \mathbf{M}\mathbf{Q} = \mathbf{0}, \nonumber \\
& \mathbf{Q} \succeq \mathbf{0}, \;\;\; \mathbf{M} \succeq \mathbf{0},\;\;\; \mathbf{tr}(\mathbf{Q}) \leq P,
\end{align}
where the equality constraints in (\ref{eq:KKT_nosk}) are the complementary slackness conditions. Note that, (\ref{eq:eq_MIMO_NOSK}) is not a concave problem in general. Thus, KKT conditions are not sufficient for optimality. Yet, since the constraint set is compact and convex and the objective function is continuous, KKT conditions are necessary for optimality. Hence, we proceed by finding the stationary points of the gradient of the dual Lagrange problem in the direction of $\mathbf{Q}$ and obtain the stationary points that solve the KKT conditions. By inspecting the objective function at these points, the global optimum can be identified. To identify the stationary points of the Lagrangian (\ref{eq:lagrange_nosk}), we get its gradient with respect to $\mathbf{Q}$ as follows:

\begin{align}
\bigtriangledown_{\mathbf{Q}}\mathcal{L} = &-\left[\mathbf{I}_{N_a} + \dfrac{1}{\sigma_b^2} \mathbf{W}_b \mathbf{Q}  \right]^{-1}\dfrac{\mathbf{W}_b}{\sigma_b^2}  + \lambda \mathbf{I}_{N_a} -\mathbf{M} \nn \\
&+ \dfrac{ \eta \gamma_w}{\sigma_w^2} \left(\left[\mathbf{I}_{N_a} + \dfrac{\gamma_w}{\sigma_w^2} \mathbf{Q}  \right]^{-1}\right)\nn \\
= &-\left[\mathbf{I}_{N_a} + \dfrac{1}{\sigma_b^2} \mathbf{W}_b \mathbf{Q}  \right]^{-1}\dfrac{\mathbf{W}_b}{\sigma_b^2}  + \lambda \mathbf{I}_{N_a} -\mathbf{M}  \nn \\
&+ \eta  \left[\dfrac{\sigma_w^2}{\gamma_w}\mathbf{I}_{N_a} +  \mathbf{Q}  \right]^{-1}
\end{align}
Assume, without loss of generality, that $\mathbf{Q} \succ \mathbf{0}$, then, from $\mathbf{M}\mathbf{Q} = \mathbf{0}$ it follows that $\mathbf{M} = \mathbf{0}$. Now, from $\bigtriangledown_{\mathbf{Q}}\mathcal{L} = 0$ we obtain
\begin{align}
\lambda \mathbf{I}_{N_a} = &\left[\sigma_b^2 \mathbf{W}_b^{-1}  + \mathbf{Q}  \right]^{-1}- \eta  \left[\dfrac{\sigma_w^2}{\gamma_w}\mathbf{I}_{N_a} +  \mathbf{Q}  \right]^{-1}
\end{align}
Since we know that $\mathbf{Q}_*$ and $\mathbf{W}_b$ have the same eigenvectors, hence, the eigenvalues of $\mathbf{Q}_*$, $\Lambda_{ii}$, can be found from
\begin{align}
\lambda = &(\sigma_b^2 \lambda_i^{-1}(\mathbf{W}_b) + \Lambda_{ii})^{-1}- \eta  \left(\dfrac{\sigma_w^2}{\gamma_w} + \Lambda_{ii}\right)^{-1},
\end{align}
as required. $\blacksquare$
\section{Proof of Theorem \ref{thm:srl_mimo_nosk}}
\label{app:srl_mimo_nosk}
\textbf{Achievability.} We observe that
\begin{eqnarray}
\label{eq:log_acheiv}
\log \abs{\mathbf{I}_{N_a}+ \dfrac{1}{\sigma_w^2}\mathbf{W}_w \mathbf{Q}} &\stackrel{(a)}{\leq}& \mathbf{tr} \left\{\dfrac{1}{\sigma_w^2}\mathbf{W}_w \mathbf{Q}\right\}\nn \\
 &=& \sum_{i=1}^{M} (\dfrac{\gamma_w \Lambda_{ii}}{\sigma_w^2}) \nn \\
 &\stackrel{(b)}{\leq}&  \dfrac{ P_{th} M \gamma_w \Lambda_{ii}}{N \sigma_w^2}
\end{eqnarray} 
where $(a)$ follows from the inequality $\log \abs{A} \leq \mathbf{tr}\left\{A - \mathbf{I}\right\}$, $(b)$ follows since the RHS is maximized for $\Lambda_{ii} = P_{th}/N$ for all $i$. For Alice to ensure that the constraint in \eqref{eq:eq_MIMO_NOSK} is satisfied, she needs the RHS of (\ref{eq:log_acheiv}$(b)$) to be less than or equal $2\delta^2/n$. Thus, Alice needs
\begin{align}
\label{eq:p_th_nosk}
P_{th} \leq \dfrac{2 N \sigma_w^2 \delta^2 }{\gamma_w n M }.
\end{align} 
Now let Alice set $P_{th}$ for (\ref{eq:p_th_nosk}) to be met with equality. Given that choice of $P_{th}$, the LPD constraint is met and the rest of the problem is that of choosing the input distribution to maximize the achievable rate. The solution of the problem is then the SVD precoding with conventional water filling \cite{telatar1999capacity,tse2005fundamentals} as follows:
\begin{align}
\label{eq:diag_q_general_nosk}
\Lambda_{ii} = \left\{
	\begin{array}{ll}
		   (\mu - \sigma_b^2 \lambda_i^{-1}(\mathbf{W}_b))^+& \mbox{for } 1 \leq i \leq N \\
		0& \mbox{for } N < i \leq N_a, \
	\end{array}
 \right.        
\end{align} 
where $\lambda_i$ is the $i^{th}$ non zero eigenvalue of $\mathbf{W}_b$. Further,  $\mu$ is a constant chosen to satisfy the power constraint $\mathbf{tr}\{\Lambda\} = P_{th}$. Accordingly, the following rate is achievable over Alice to Bob channel:
\begin{align}
\label{eq:cap_nosk}
R_{pd}(\delta) &= \sum_{i=1}^{N} \log \left(1+ \dfrac{(\mu - \sigma_b^2 \lambda_i^{-1}(\mathbf{W}_b))^+\lambda_i(\mathbf{W}_b)}{\sigma_b^2}\right) \nn \\
&= \sum_{i=1}^{N} \left(\log \left(\dfrac{\mu \lambda_i(\mathbf{W}_b)}{\sigma_b^2}\right)\right)^+.
\end{align}

However, it is technically difficult to expand \eqref{eq:cap_nosk} to check the applicability of the square-root law. Therefore, we obtain an achievable rate assuming that Alice splits $P_{th}$ equally across active eigenmodes of her channel to Bob. Note that, this rate is indeed achievable since it is less than or equal to \eqref{eq:cap_nosk}. Also note that, when Alice to Bob channel is well conditioned, the power allocation in \eqref{eq:diag_q_general_nosk} turns into equal power allocation. Then, the following rate is achievable:
\begin{align}
\label{eq:R_acheiv_nosk}
R(\delta) = \sum_{i=1}^{N} \log \left(1+ \dfrac{2 \sigma_w^2 \delta \lambda_i(\mathbf{W}_b)}{\sigma_b^2\gamma_w n M} \right).
\end{align}
 Now suppose that Alice uses a code of rate $\hat{R} \leq R(\delta)$, then, Bob can obtain
 \begin{align}
 n \hat{R} &\leq n  \sum_{i=1}^{N} \log \left(1+ \dfrac{2\sigma_w^2 \delta^2 \lambda_i(\mathbf{W}_b)}{\sigma_b^2\gamma_w n M} \right) \nn \\
 &\stackrel{(a)}{\leq}   \sum_{i=1}^{N} \dfrac{\sqrt{2}\sigma_w^2 \delta^2 \lambda_i(\mathbf{W}_b)}{M \sigma_b^2\gamma_w \ln 2}
 \end{align}
bits in $n$ independent channel uses, where $(a)$ follows from $\ln(1+x) \leq x$, and note that the inequality is met with equality for sufficiently large $n$. Now, assume that $\lambda_i(\mathbf{W}_b) = \lambda_b$ for all $1 \leq i \leq N$, i.e., Bob's channel is well conditioned. Then,
Bob can obtain
 \begin{align}
 n R(\delta) \leq \dfrac{N\sqrt{2}\sigma_w^2 \delta^2 \lambda_b}{M \sigma_b^2\gamma_w \ln 2}
 \end{align}
bits in $n$ independent channel uses since the inequality is met with equality for sufficiently large $n$.$\blacksquare$

\textbf{Converse.} To show converse, we assume the most favorable scenario for Alice to Bob channel when $\mathbf{H}_b$ is well conditioned. Then, Alice will split her power equally across active eigenmodes of her channel to Bob. Now, 
\begin{eqnarray}
\label{eq:log_nosk}
\log \abs{\mathbf{I}_{N_a}+ \dfrac{1}{\sigma_w^2}\mathbf{W}_w \mathbf{Q}} &\stackrel{(a)}{\geq}& \mathbf{tr} \left\{\mathbf{I}_{N_a} - \left(\mathbf{I}_{N_a}+\dfrac{1}{\sigma_w^2}\mathbf{W}_w \mathbf{Q} \right)^{-1}\right\}\nn \\
 &=& \sum_{i=1}^{N_a} \left( 1 - \dfrac{1}{1+ \dfrac{\gamma_w P_{th}}{N \sigma_w^2}}\right) \nn \\
 &=& \left(\dfrac{ M \gamma_w P_{th}}{N \sigma_w^2 + \gamma_w P_{th} }\right)
\end{eqnarray} 
where $(a)$ follows from the inequality $\log \abs{A} \geq \mathbf{tr}\left\{\mathbf{I} - A^{-1} \right\}$. Following the same steps as in the achievability proof, we can insure that
\begin{align}
\label{eq:p_th_converse_nosk}
P_{th} \leq  \dfrac{2 \xi N \xi \sigma_w^2 \delta^2 }{\gamma_w n M },
\end{align} 
where $\xi = \dfrac{nM }{n M - 2 \delta^2 } > 1$, otherwise, Alice can not meet the LPD constraint. Now let Alice set $P_{th}$ equals to the RHS of (\ref{eq:p_th_converse_nosk}). Then, we can verify that:
\begin{align}
\label{eq:R_aconverse_nosk}
C_{pd}(\delta) \leq \sum_{i=1}^{N} \log \left(1+ \dfrac{2 \xi \sigma_w^2 \delta^2 \lambda_b}{\sigma_b^2\gamma_w n  M } \right).
\end{align}
 Now suppose that Alice uses a code of rate $ R(\delta) \leq C_{pd}(\delta)$, then, Bob can obtain
 \begin{align}
 n R(\delta) &\leq n  N \log \left(1+ \dfrac{2\xi\sigma_w^2 \delta^2 \lambda_b}{\sigma_b^2\gamma_w n M } \right) \nn \\
 &\leq    \dfrac{2 N  \xi\sigma_w^2\delta^2 \lambda_b}{ M \sigma_b^2\gamma_w \ln 2}
 \end{align}
bits in $n$ independent channel uses since the inequality is met with equality for sufficiently large $n$. $\blacksquare$


\end{appendices}

\end{document}